\documentclass[lettersize,journal]{IEEEtran}
\usepackage{amsmath,amsfonts}
\usepackage{algorithmic}
\usepackage{array}
\usepackage{hyperref}
\usepackage{caption}
\usepackage{subfig}
\usepackage{textcomp}
\usepackage{stfloats}
\usepackage{url}
\usepackage{verbatim}
\usepackage{makecell}
\usepackage{graphicx}
\usepackage{hyperref}
\usepackage{array}
\usepackage{multirow}
\usepackage{amsmath}
\usepackage{amsthm}
\usepackage{amssymb}
\usepackage{etoolbox}
\usepackage{algorithmic}
\usepackage{algorithm}
\usepackage{booktabs} 
\usepackage{graphicx}
\usepackage{multirow}
\usepackage{array} 
\usepackage{xspace}
\usepackage{newtxtext,newtxmath}
\usepackage{balance}
\usepackage{cite}
\allowdisplaybreaks

\newtheorem{myTheo}{Theorem}

\newtheorem{myLem}{Lemma}
\newtheorem{myRmk}{Remark}

\def\BibTeX{{\rm B\kern-.05em{\sc i\kern-.025em b}\kern-.08em
		T\kern-.1667em\lower.7ex\hbox{E}\kern-.125emX}}

\begin{document}
	\title{Semantic-aware Token Selection and Resource Optimization for Communication-efficient Split Federated Fine-tuning in Edge Intelligence} 
    
	\author{
		Xianke Qiang, Zheng Chang,~\IEEEmembership{Senior~Member,~IEEE,} Geyong~Min,~\IEEEmembership{Senior Member,~IEEE}
		\thanks{X. Qiang and Z. Chang are with School of Computer Science and Engineering, University of Electronic Science and Technology of China, Chengdu 611731, China. G. Min is with Department of Computer Science, University of Exeter, Exeter, EX4 4QF, U.K.}
	}
	\maketitle

	\begin{abstract}
		  Deploying large Transformer-based vision models on resource-limited mobile devices at network edge is severely constrained by hardware limitations and dynamic wireless environments. While federated learning (FL) enables collaborative training without sharing raw data, strictly local fine-tuning of such massive models remains computationally prohibitive for edge devices. Split federated learning (SFL) alleviates this burden by offloading deep layers to the edge server, yet it suffers from heavy communication overhead when transmitting high-dimensional activation tokens. To address this bottleneck, we propose ST-SFLora, a semantic token-based split federated LoRA fine-tuning framework. We introduce a new metric, \emph{Semantic Transmission Efficiency} (STE), to balance semantic retention and transmission cost. Based on STE, we formulate a joint resource optimization problem that dynamically determines token selection, uplink bandwidth allocation, and transmit power under latency and energy constraints. The resulting mixed-integer nonconvex problem is efficiently solved via an alternating algorithm. Experiments on multiple benchmarks demonstrate that ST-SFLora achieves the lowest client-side resource consumption among baselines while delivering a favorable trade-off between communication efficiency and model performance.
	\end{abstract}
	
	\begin{IEEEkeywords}
		  split federated learning, resource allocation, vision transformer, fine-tuning, token selection.
	\end{IEEEkeywords}

	\section{Introduction}	

    Driven by rapid advances in Aritifical Intelligence (AI), Transformer-based foundation models such as BERT~\cite{devlin2019bert}, ViT~\cite{dosovitskiy2020image}, and GPT~\cite{achiam2023gpt} have achieved state-of-the-art performance in natural language processing and computer vision, demonstrating strong scalability and generalization across modalities. Their success is largely enabled by the pre-training and fine-tuning paradigm, which allows large models to be adapted efficiently to downstream tasks\cite{bai2024federated}. At the same time, the rapid proliferation of intelligent devices in mobile edge computing (MEC) systems is generating large volumes of decentralized and privacy-sensitive data at the wireless edge~\cite{lin2024splitlora}. Relying on cloud-based centralized learning to collect raw data from edge devices raises critical concerns, including privacy exposure and excessive bandwidth consumption \cite{10835069}. Federated learning (FL)\cite{cai2023efficient,mcmahan2017communication} addresses these issues by keeping data local and aggregating model updates on the server. However, fine-tuning Transformer-based large models on heterogeneous edge devices remains highly challenging. Even with parameter-efficient techniques such as LoRA, the memory footprint and communication demand of these models still exceed the capabilities of typical mobile devices\cite{li-etal-2025-mobilora}.\par
    To overcome the limitations of on-device fine-tuning, Split Learning (SL)\cite{vepakomma2018split} and Split Federated Learning (SFL)\cite{thapa2022splitfed} offload most computation from clients to the server by partitioning the model between the two sides. However, when applied to large Transformer architectures, these approaches introduce new communication bottlenecks. Traditional SFL requires transmitting client-side weights, intermediate activations, and activation gradients in every iteration\cite{thapa2022splitfed,10839234}. Among these, uploading high-dimensional activations dominates the communication cost\cite{qiang2025deploying}. For example, the activations produced by a ViT-B/16 model for a single batch can exceed tens of megabytes\texorpdfstring{\textsuperscript{1}}{}, which is particularly problematic in mobile scenarios where bandwidth fluctuates with channel conditions and connectivity is frequently intermittent. To reduce this communication overhead, existing approaches typically rely on bit-level compression techniques such as quantization \cite{9611373,qiang2025deploying} and sparsification\cite{10256151}. \par
    \begin{figure}[t]
        \centering
        \includegraphics[width=0.7\linewidth]{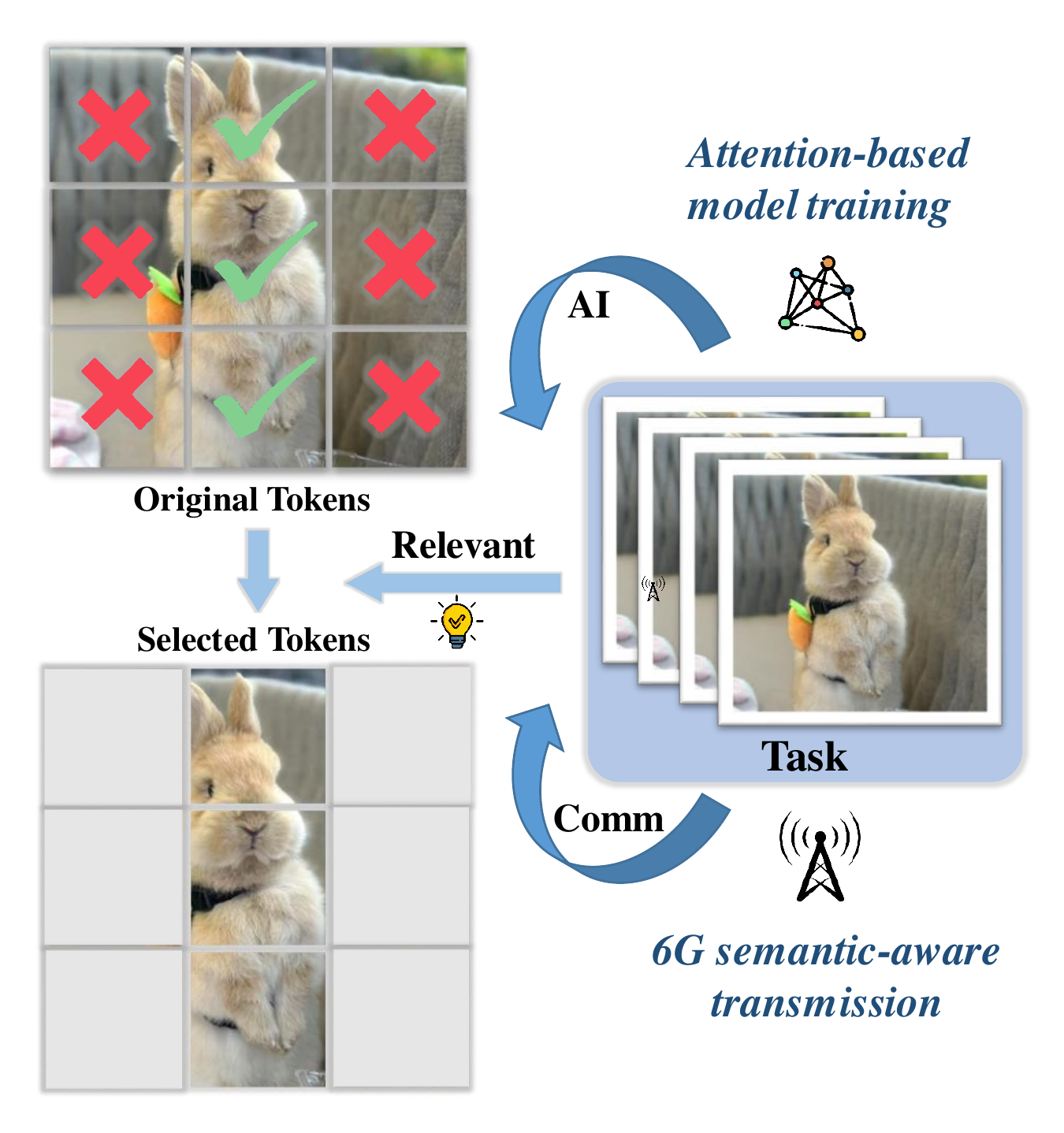}
        \caption{Attention-based semantic-aware token selection for 6G. The attention mechanism selectively transmits the most informative tokens.}
        \label{fig:task-related-tokens}
    \end{figure}
    \footnotetext[1]{Calculated based on a batch size of 64, a sequence length of 197 (14$\times$14 patches plus one class token), and a hidden dimension of 768 with 32-bit floating-point precision. The total size is $64 \times 197 \times 768 \times 4$ bytes $\approx$ 37 MB.}

    \begin{figure*}[t]
        \centering
        \includegraphics[width=1.0\linewidth]{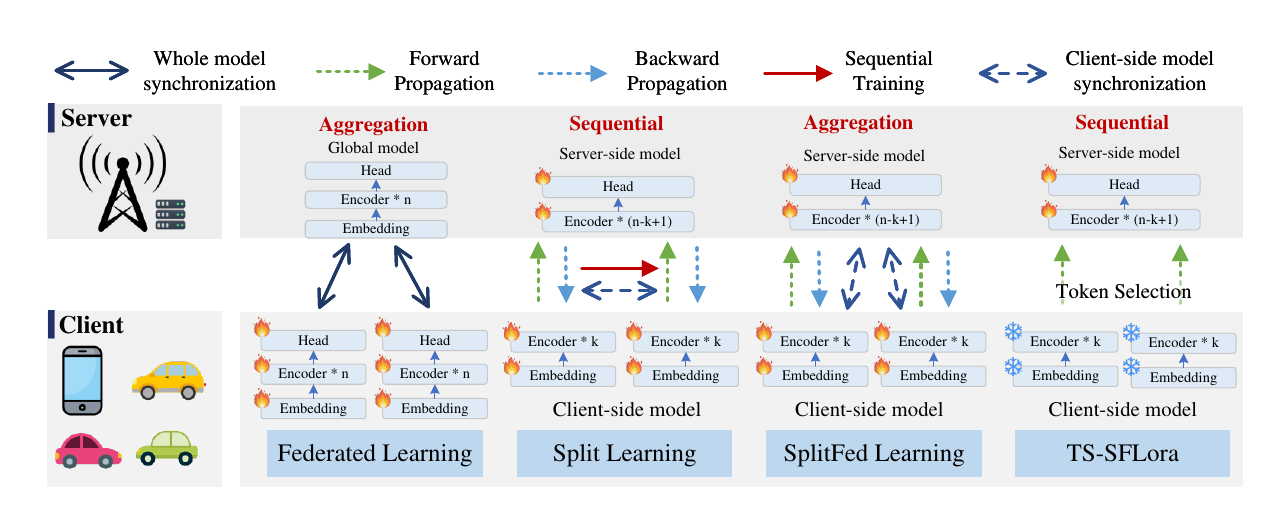}
        \caption{Comparison of Distributed Learning Architectures.}
        \label{fig:architecturecompare}
    \end{figure*}

    As wireless networks advance toward 6G, communication paradigms are shifting from bit-level transmission to task-oriented semantic transmission, where only information relevant to task performance is conveyed and redundant content is suppressed\cite{10054381}. This evolution aligns naturally with attention mechanism of Transformer-based SL/SFL, in which activations are inherently represented as sequences of semantic tokens rather than raw bits. As shown schematically in Fig.~\ref{fig:task-related-tokens}, the attention maps can highlight task-relevant patches and ignore background regions, so that only the most informative tokens need to be sent over the 6G link. These tokens contribute unequally to the final prediction \cite{zeng2022not}. This leads to a natural question: \textit{can we compress activations at the semantic token level rather than at the bit level?}\par
    Although semantic-aware token-level compression holds promise, applying it to large-scale SFL in MEC environments presents three fundamental challenges:
    \begin{itemize}
        \item Unlike FL, which exchanges gradients or parameter updates at the bit level, SL/SFL transmits intermediate activations that, in Transformer models, are naturally tokenized and semantically structured. This enables semantic cues to be extracted at the token level rather than treating activations as raw bitstreams. The challenge is to design an SFL architecture that exposes these semantic cues directly from the backbone while keeping clients lightweight and avoiding additional semantic encoders, scoring networks, or auxiliary modules.
        \item Since not all tokens contribute equally to the downstream task, transmitting full activation sequences is inefficient. The challenge is to design a principled selection mechanism that leverages the backbone’s attention scores to identify and retain task-relevant tokens, while remaining robust under heterogeneous data distributions and mobile conditions.
        \item Token selection introduces a trade-off between semantic information retention and communication overhead. The challenge is to formulate a semantic-aware metric and an optimization problem that balance this trade-off by jointly determining the number of transmitted tokens and the uplink bandwidth and transmit power allocation, under latency and energy constraints.
    \end{itemize}

    To tackle these challenges, we propose ST-SFLora, a semantic token-based split federated LoRA fine-tuning framework with three key designs. First, we adopt a one-way uplink SFL architecture in which clients only run lightweight forward passes and upload selected token activations, while avoiding gradient feedback and any extra semantic/scoring modules on devices; this keeps the client logic simple, improves scalability, and enables elastic participation under unreliable mobile connectivity. Second, we exploit the Transformer’s native attention signals to build a semantic-aware token selection strategy, reusing attention weights to identify the most informative tokens without modifying the backbone or introducing noticeable computational overhead. Finally, we introduce a system-level metric, Semantic Transmission Efficiency (STE), to quantify the trade-off between semantic preservation and communication cost under token compression, and use it to couple token selection with bandwidth and power allocation for unified communication resource optimization in ST-SFLora.\par

    To the best of our knowledge, this work represents the first attempt to integrate semantic-aware token selection with split federated LoRA fine-tuning for edge intelligence. The main contributions of this paper are summarized as follows:
    
    \begin{itemize}
        \item We propose ST-SFLora, a decoupled training framework where clients perform only lightweight forward propagation to extract activation tokens, while all trainable LoRA parameters reside on the server. A key design is to select task-relevant tokens before uplink transmission to balance semantic retention and communication cost, leading to a smaller activation payload in mobile scenarios.
        
        \item We define a novel metric, Semantic Transmission Efficiency (STE), to fundamentally quantify the system-level trade-off between semantic information retention and communication latency. Unlike traditional throughput metrics, STE explicitly characterizes the effective semantic information delivered per unit of time.
        
        \item We formulate a joint optimization problem that coordinates token selection, uplink transmit power, and bandwidth allocation under strict latency and energy constraints. To tackle the resulting mixed-integer nonconvex program, we develop an efficient alternating optimization algorithm.
    \end{itemize}
    
    The remainder of this paper is organized as follows. Section \ref{section2} reviews the background and related work. Section \ref{section3} presents the workflow and delay–energy analysis of ST-SFLora. Section~\ref{task-relevant-token-selection} presents the client selection and token selection strategies. Section \ref{section6} presents the new metric STE and formulates the optimization problem, and Section \ref{section7} provides its solution. Section \ref{section8} reports the performance evaluation. Section \ref{section9} concludes the paper.

    \section{Related Works} 
    \label{section2}
    FL enables collaborative model training without sharing raw data by aggregating client-side updates on a central server \cite{mcmahan2017communication, cai2023efficient}. Although effective for privacy preservation, FL becomes challenging to apply to large Transformer-based models because local fine-tuning demands substantial computation and memory on heterogeneous mobile devices. SL \cite{vepakomma2018split} addresses device constraints through sequential collaborative training: each client computes the early layers of a shared model, uploads the resulting activations, and receives gradients from the server before the next client proceeds. SFL \cite{thapa2022splitfed} combines the model-partitioning idea of SL with the aggregation mechanism of FL, enabling clients to execute the early layers in parallel and upload activations, while the server completes the remaining forward–backward computation and aggregates updates across clients. The workflows of these representative paradigms are summarized in Fig.~\ref{fig:architecturecompare}.\par
    Recent works have explored adaptive and resource-aware variants of SFL. HSFL \cite{9923620} integrates the parallel update mechanism of FL with the model-splitting structure of SL and incorporates a MAB-based user selection strategy to improve accuracy under heterogeneous UAV networks. CPSL \cite{wu2023split} proposes a parallel-then-sequential split learning scheme that jointly optimizes cut-layer selection, device clustering, and spectrum allocation to reduce training latency in wireless networks. In vehicular edge intelligence, ASFV \cite{10714368} dynamically selects cut layers based on mobility and resource availability. SFL has also been extended to large models: SplitLoRA \cite{lin2024splitlora} integrates LoRA fine-tuning into SFL to reduce trainable parameters, while SplitFrozen~\cite{ma2025splitfrozen} freezes most client-side layers and also keeps part of the server-side backbone frozen, while only updating lightweight auxiliary modules to accommodate device heterogeneity and accelerate LLM fine-tuning.\par
    Communication-efficient training has been widely studied in distributed learning. Bit-level compression techniques such as gradient quantization~\cite{9611373,9277666,liu2023communication,10038639,10621361,10542529} and sparsification~\cite{10256151,10026255} reduce gradient communication in FL. In SL and SFL, activation compression \cite{qiang2025deploying,10791300,zheng2023reducing} has received increasing attention, as transmitting intermediate features often dominates uplink costs. As wireless networks evolve toward 6G, task-oriented semantic communication \cite{10054381,10538233} promotes transmitting only task-relevant information. However, existing SL/SFL frameworks still rely on bit-level compression and have not exploited the semantic structure of Transformer activations—whose representation is naturally tokenized and semantically structured—for selective activation transmission or communication optimization under mobility-induced channel variation.\par
    Existing SL and SFL frameworks struggle to support efficient fine-tuning of large Transformer models on mobile edge devices, due to heavy communication overhead and the need for bidirectional exchanges of large activations. Prior compression methods operate at the bit level and ignore the token-level semantic structure of Transformer activations. In mobile scenarios, mobility-induced channel variations further worsen the communication bottleneck, but are rarely considered together with wireless resource constraints and semantic relevance. These limitations motivate the proposed ST-SFLora framework, which (i) reduces client–server interactions, (ii) leverages token-level semantics to guide activation transmission, and (iii) jointly optimizes token selection and wireless communication resources.

    \section{ST-SFLora: Overview and Delay-Energy Analysis}
    \label{section3}
    In this section, we first present the system overview of ST-SFLora, and then describe its training workflow along with the corresponding delay-energy analysis.
    
    \subsection{Architecture Overview}
    The proposed ST-SFLora framework is tailored for collaborative fine-tuning over MEC networks. The system comprises an edge server and a set of edge devices $\mathcal{M}=\{1,2,\dots,M\}$. Each device $m\in\mathcal{M}$ stores a private dataset $\mathcal{D}_m=\{(\mathcal{X}_m,\mathcal{Y}_m)\}$ with $D_m$ labeled samples. Training proceeds over communication rounds $t\in\{1,2,\dots,T\}$. Due to limited wireless bandwidth and time-varying channels, only a subset of clients $\mathcal{U}_t\subseteq\mathcal{M}$ participates in round $t$. The global model is partitioned at cut layer $e$: shallow layers run on the device, while the remaining layers are hosted on the edge server.

	\subsection{Workflow and Delay-Energy Analysis} The overall workflow of the proposed ST-SFLora framework consists of six phases, as illustrated in Fig.~\ref{fig:workflow}. The training procedure is summarized in Alg.~\ref{algo1}. The detailed process of each phase is described as follows.
	\begin{figure}[t]
		\centering
		\includegraphics[width=1.0\linewidth]{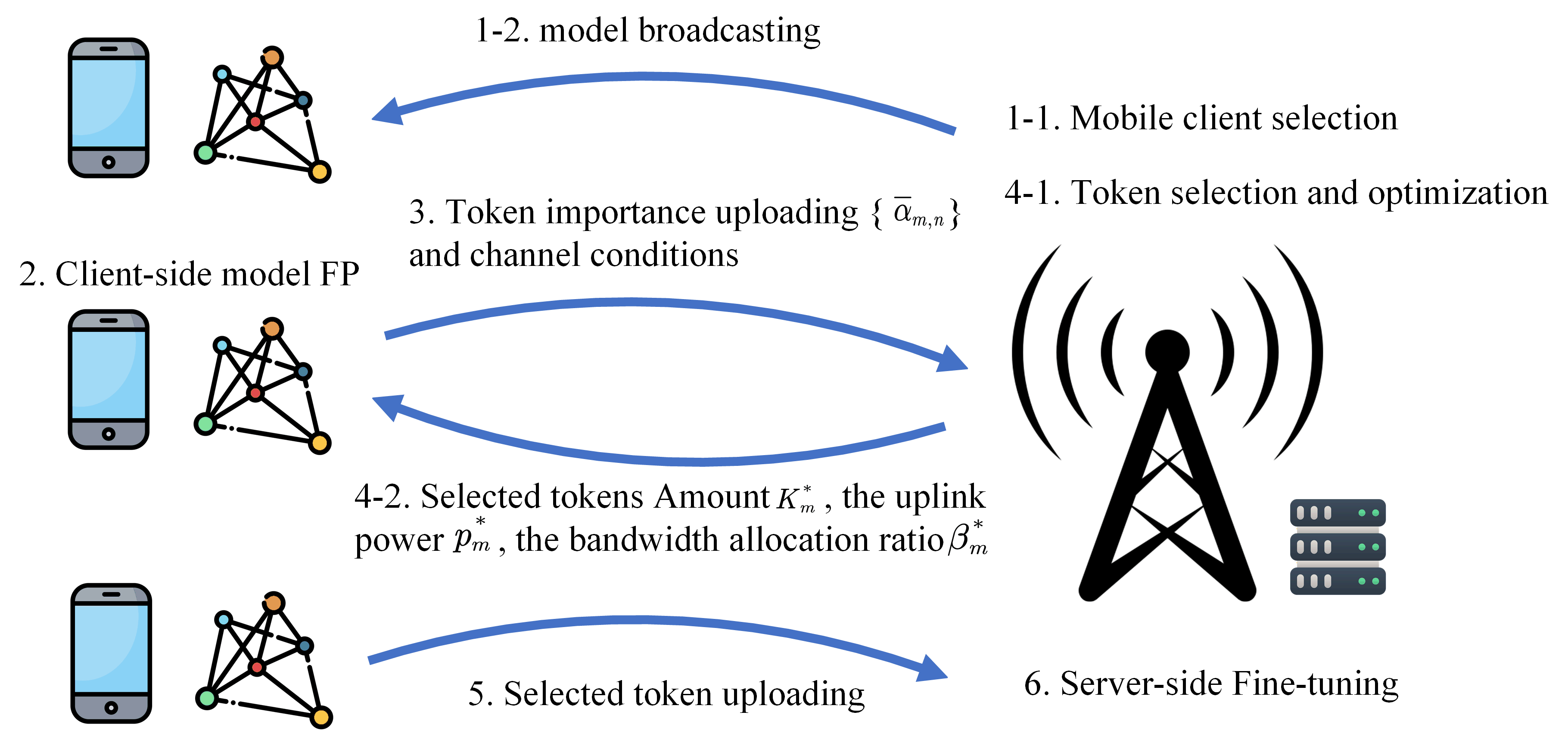}
		\caption{The workflow of the ST-SFLora system.}
		\label{fig:workflow}
	\end{figure} 

	\subsubsection{Client Selection and Model Broadcasting} At the beginning of each global round~$t$, the edge server collects the channel state information (CSI) and status of devices within its communication range and performs mobility-aware client selection. The resulting active client set is denoted by $\mathcal{U}_t \subseteq \mathcal{M}$, and the server broadcasts the pre-trained client-side model parameters $\boldsymbol{\omega}^{c}$ to all clients $m \in \mathcal{U}_t$.\par
    The control signaling for client selection is negligible compared with the transmission of activation tokens and is therefore ignored. The downlink broadcasting rate is $R^{DL} = W_{\mathrm{tot}} \log_2\!\left(1 + \frac{p_s h_{\min}}{N_0 W_{\mathrm{tot}}}\right)$, where $h_{\min} = \min_{m\in\mathcal{M}} h_m$ represents the weakest channel gain among all participating clients, $p_s$ is the downlink power of sever, $W_{\mathrm{tot}}$ is the overall bandwidth. Accordingly, the broadcasting delay for a model of size $s_0$ bits can be expressed as
    	\begin{equation}
    		T^{DL} = \frac{s_0}{R^{DL}}.
    	\end{equation}

    \subsubsection{Client-Side Feature Extraction and Tokenization}
       Each participating client $m \in \mathcal{U}_t$ performs forward propagation to extract local feature representations using the client-side model. Specifically, given a local mini-batch $\mathcal{B}_{m}=\{(x_b,y_b)\}_{b=1}^{B}$ sampled from the private dataset $\mathcal{D}_{m}$, the client computes the intermediate activations as: $A_{m}=\ell\!\big(\boldsymbol{\omega}^{c},\,\mathcal{B}_{m}\big)$, where $\ell(\cdot)$ denotes the local forward function that maps the input mini-batch $\mathcal{B}_{m}$ to a sequence of tokens $A_{m}\in\mathbb{R}^{B\times (N+1)\times D}$.\texorpdfstring{\textsuperscript{2}}{}\par
        Since all clients execute this phase in parallel, we analyze the local computation delay for an individual client. Let $\gamma^{F}_{c} =\sum_{j=1}^{e}\rho_{j}$ represent the total computational complexity (in FLOPs) of the client-side model per sample, where $\rho_{j}$ denotes the workload of the $j$-th layer and $e$ is the total number of client-side layers. Consequently, the computation latency for client $m$ to process a mini-batch of size $B$ is formulated as:
            \begin{equation}
                T^{F}_{m}=\frac{B\,\gamma^{F}_{c}}{f_{m}\,C_{m}\,D^{U}_{m}},
            \end{equation}
        where $f_{m}$ is the GPU clock frequency of client $m$, $C_{m}$ is the number of computing cores, and $D^{U}_{m}$ represents the floating-point operations per cycle per core.
    
         \footnotetext[2]{Given an image $I \in \mathbb{R}^{H \times W \times C}$, the patch-embedding layer first splits $I$ into $N$ non-overlapping patches and flattens them into $I_p \in \mathbb{R}^{N \times (P^2C)}$, where $(H, W)$ represents the height and width of the input image, $(P, P)$ is the resolution of each image patch, $C$ denotes the number of channels, and $N = HW / P^2$ is the number of image tokens. The flattened patches $I_p$ are then mapped to $A^0 = [a_1, a_2, \ldots, a_N] \in \mathbb{R}^{N \times D}$ using a trainable linear projection. Finally, a learnable \texttt{\texttt{[CLS]}} token $a_0$ is prepended to form $[a_0, a_1, \ldots, a_N] \in \mathbb{R}^{(N+1)\times D}$, which serves as the input to the Transformer layers.}

	\begin{figure}[t]
		\centering
		\includegraphics[width=0.82\linewidth]{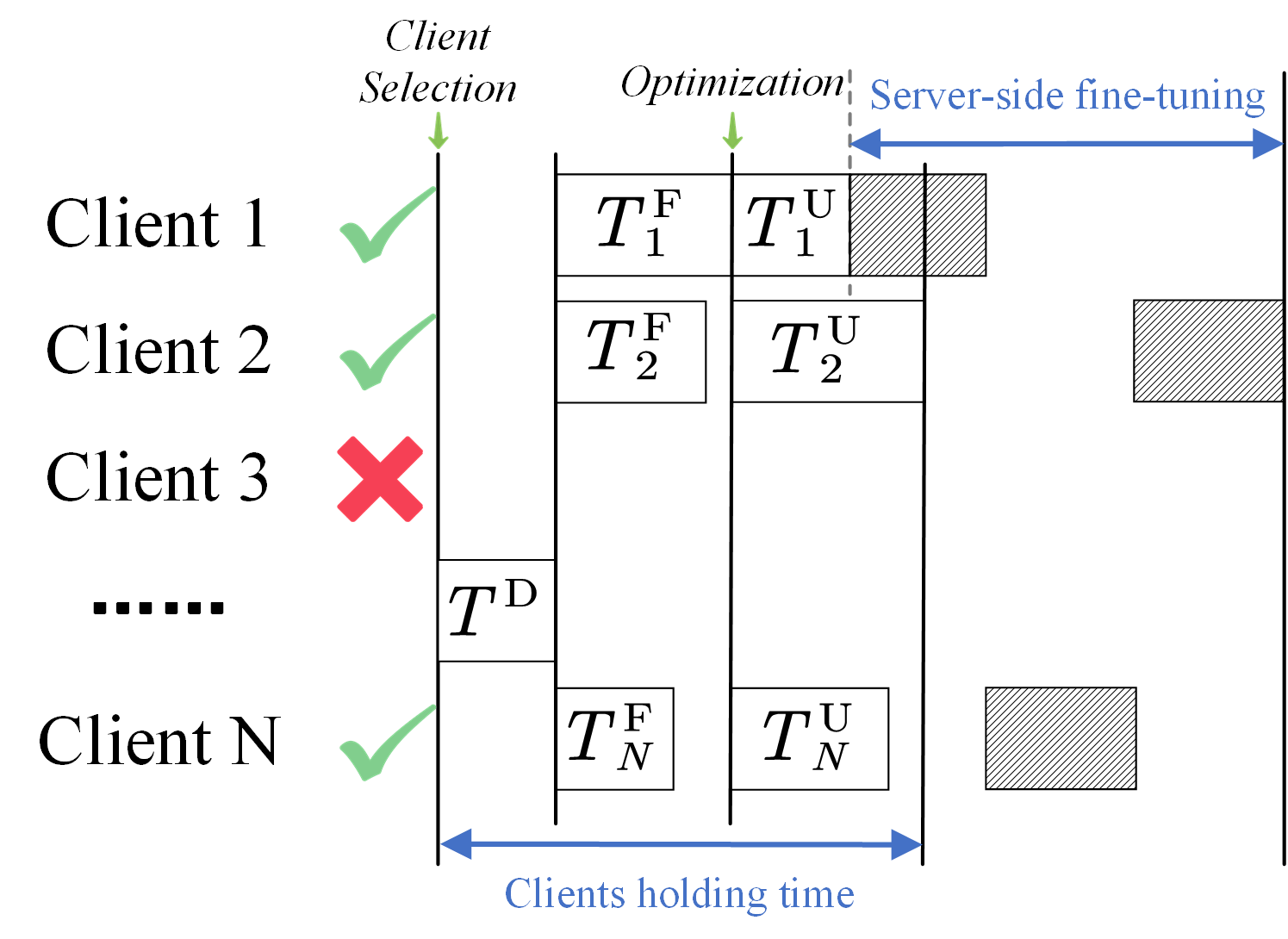}
		\caption{The timeline of ST-SFLora.}
		\label{fig:timeline}
	\end{figure} 
    \subsubsection{Token Importance Uploading} For each sample in the mini-batch, the final client-side Transformer block outputs a set of tokens accompanied by attention scores that indicate their relative semantic importance. To characterize the global attention distribution, the client sorts the tokens of each sample in descending order of importance and sums the attention values rank-wise across the batch. This operation yields a batch-level token importance vector $\bar{\boldsymbol{\alpha}}_{m} = [\bar{\alpha}_{m,1}, \bar{\alpha}_{m,2}, \ldots, \bar{\alpha}_{m,N}]$. Subsequently, each client transmits this batch-level token importance vector to the server. Since these vectors consist of lightweight scalar quantities, the transmission latency and energy consumption of this phase are considered negligible and are ignored in the subsequent analysis.

    \subsubsection{Optimization}Utilizing the received CSI and the batch-level token importance vectors, the server performs joint optimization of the token budget $\boldsymbol{K}$, uplink transmit power $\boldsymbol{p}$, and bandwidth allocation $\boldsymbol{W}$. Upon solving the optimization problem, the server distributes the optimal parameters $\{\boldsymbol{K}^*, \boldsymbol{p}^*, \boldsymbol{W}^*\}$ back to the selected clients. Similar to the previous phase, the computational time for this optimization and the downlink latency for broadcasting these control parameters are omitted from the delay-energy analysis, as they are minimal compared to the training and activation transmission costs.
	
    \subsubsection{Selected Token Uploading} Upon receiving the optimized parameters, each selected client applies the semantic-aware token selection strategy based on the assigned token budget $K_m^*$ and communication parameters $(W_m^*, p_m^*)$. Only the refined activations $A_m^{\mathrm{ref}} \in \mathbb{R}^{B \times (K_m^*+2) \times D}$ are transmitted to the server over the allocated bandwidth. This selective transmission effectively reduces communication overhead while preserving sufficient semantic information for server-side fine-tuning.\par
    To upload the selected tokens, the total bandwidth $W_{\mathrm{tot}}$ is partitioned among the active clients. Let $W_m$ denote the bandwidth allocated to client $m$ in round $t$, where $\sum_{m\in\mathcal{U}_t}W_m \leq W_{\mathrm{tot}}$. The achievable uplink transmission rate for client $m$ is given by:
        \begin{equation}
            R_{m}^{U L}= W_m \log_2 \left(1+\frac{p_{m}^* h_m}{N_0 W_m}\right),
        \end{equation}
    where $p_{m}^*$ is the optimized transmit power, $h_m$ is the channel gain, and $N_0$ is the background noise power spectral density. The data size of the uploaded tokens for client $m$ depends on the batch size $B$, the number of selected tokens $K_m^*$, and the embedding dimension $D$. It is calculated as:
        \begin{align}
            S_m &= B \times (K_{m}^*+2) \times D \times q_0 \quad (\mathrm{bits}),
            \label{eq:s_mt}
        \end{align}
    where $q_0$ denotes the bit-width per element (e.g., $q_0=32$ for single-precision floating-point without quantization). Consequently, the uplink communication latency and energy consumption are formulated as:
        \begin{align}
            T_{m}^{\mathrm{U}} &= \frac{S_{m}}{R_{m}^{UL}}, \quad E_{m}^{\mathrm{U}} = p_{m}^* T_{m}^{\mathrm{U}}.
        \end{align}
   
    \subsubsection{Server-Side LoRA-Based Fine-Tuning} 
        Upon receiving the refined activations from $\mathcal{U}_t$, the server sequentially fine-tunes the LoRA adapter $\Delta \boldsymbol{\omega}^{s}$ while keeping the backbone frozen. For each client $m \in \mathcal{U}_t$, the uploaded $A_m^{\mathrm{ref}}$ is used to compute the loss $L$ and update the parameters as:
            \begin{equation}
                \Delta \boldsymbol{\omega}_{m,t}^{s} \leftarrow 
                \Delta \boldsymbol{\omega}_{m-1,t}^{s} - \eta \, \nabla L,
            \end{equation}
        where $\Delta \boldsymbol{\omega}_{0,t}^{s}= \Delta \boldsymbol{\omega}_{t-1}^{s}$ inherits the global state from the previous round. This phase finalizes the learning for round $t$. Given the server's sufficient computational capacity, the associated latency and energy consumption are omitted from our analysis.
        	\begin{algorithm}[t]
                \small
        		\caption{Semantic Token-based Split Federated LoRA Fine-tuning (ST-SFLora)}
        		\label{algo1}
        		\begin{algorithmic}[1]
        			\STATE Initialize the split layer $e$ and LoRA adapter parameters.
        			\FOR{each global round $t = 1, 2, \dots, T$}
        			\STATE Server broadcasts the embedding layer and first $e$ Transformer blocks to all clients $\mathcal{U}_t$.
        			\FOR{each client $m \in \mathcal{U}_t$ in parallel}
        			\FOR{each local epoch $i = 1, 2, \dots, I$}
        			\STATE {client-side:}
        			\STATE Sample a mini-batch from local dataset $\mathcal{D}_m$.
        			\STATE Run forward propagation to obtain activations/tokens and importance scores.
        			\STATE Upload token importance, CSI, and mobility information.
        			
        			\STATE {Server-side:}
        			\STATE Collect feedback from all accessible clients.
        			\STATE Perform online joint optimization to determine $\{K_m^*, W_m^*, p_m^*\}$ parameter setting.
        			
        			\STATE {Client-side:}
        			\STATE Each selected client $m \in \mathcal{U}_t$ transmits selected refined tokens $\mathbf{A}_m^{\mathrm{ref}}$ to the server.
        			
        			\STATE {Server-side:}
        			\STATE Receive tokens and fine-tune the server-side model via LoRA-based updates.
        			\ENDFOR
        			\ENDFOR
        			\ENDFOR
        		\end{algorithmic}
        	\end{algorithm}
            
    \section{Client Selection and Token Selection}\label{task-relevant-token-selection}
    In this section, we first introduce a mobility-aware client selection scheme to ensure reliable participation within the server coverage. We then propose a semantic-aware token selection approach that combines attention-driven selection with token merging.
    \subsection{Mobility-aware Client Selection}
    Prior to the commencement of local training, the server determines which clients are eligible to participate in the current round. Due to the limited communication range of the edge server, high-mobility clients risk exiting the coverage area during the training process, leading to potential transmission failures. Consequently, not all clients can be selected in each round, and a mobility-aware selection strategy is required to ensure stable participation and reliable model updates.\par
    We define the standing time as the duration a client remains within the server's effective communication range. Let $v_m$ denote the velocity of client $m$, which is assumed constant during one round. Let $L$ represent the coverage radius of the server, and $l_m$ denote the current radial distance of client $m$ from the server. Thus, the remaining distance to the coverage boundary is $L - l_m$. Given the maximum latency constraint per iteration $\bar{t}$, the effective standing time of client $m$ is defined as:
        \begin{equation}
            T_m^{\mathrm{standing}} = \min\!\left\{\frac{L - l_m}{v_m}, \, \bar{t}\right\}.
        \end{equation}
    To ensure successful completion, the client's standing time must exceed the total time required for interaction, denoted as the holding time. As illustrated in Fig.~\ref{fig:timeline}, holding time is calculated as:
        \begin{equation}
            T_m^{\mathrm{holding}} = T_m^0 + T_m^{\mathrm{U}},
        \end{equation}
    where $T_m^0 = T^{\mathrm{D}} + T_m^{\mathrm{F}}$ represents the sum of the downlink broadcasting delay and the local forward computation time, while $T_m^{\mathrm{U}}$ denotes the uplink transmission time. A binary client selection indicator $\theta_m^t$ is then defined as:
        \begin{equation}
            \theta_m^t =
            \begin{cases}
                1, & \mathrm{if\quad} T_m^{\mathrm{holding}} \leq T_m^{\mathrm{standing}}, \\
                0, & \mathrm{otherwise},
            \end{cases}
        \end{equation}
    where $\theta_m^t = 1$ indicates that client $m$ can complete the training task within the coverage area and is selected for this round, while $\theta_m^t = 0$ indicates exclusion due to the risk of disconnection. Accordingly, the selected client set is given by:
        \begin{equation}
            \mathcal{U}_t = \{\, m \in \mathcal{M} \mid \theta_m^t = 1 \,\}.
        \end{equation}
        
    \subsection{Semantic-aware Token Selection Strategy}
        In this subsection, we present a semantic-aware token selection strategy that exploits the self-attention mechanism of Transformer architectures to suppress redundant tokens while preserving task-relevant context.
        \subsubsection{Token Selection}
        In Vision Transformers, the \texttt{\texttt{[CLS]}} token serves as a global feature aggregator, attending to all patch tokens to synthesize task-discriminative representations. Consequently, the attention weights assigned by the \texttt{\texttt{[CLS]}} token serve as a reliable proxy for task relevance. Tokens with higher attention scores typically contribute more to the final prediction, whereas those with lower scores are often redundant background features.\par
        For each client $m$, let $\mathbf{A}_m \in \mathbb{R}^{B \times (N+1) \times D}$ denoted the intermediate activation tensor obtained from client-side forward propagation, where $B$ denotes the batch size, $N$ the number of patch tokens, and $D$ the embedding dimension. The token sequence of the $b$-th image is represented as $\mathbf{A}_m^{(b)} = [\mathbf{a}_{m,0}^{(b)}, \mathbf{a}_{m,1}^{(b)}, \ldots, \mathbf{a}_{m,N}^{(b)}]$, where $\mathbf{a}_{m,0}^{(b)}$ is the \texttt{\texttt{[CLS]}} token and $\mathbf{a}_{m,n}^{(b)}$ ($n=1,\ldots,N$) are patch tokens. In the final client-side Transformer block, the self-attention map is computed as
        	\begin{equation}
        		\mathbf{Att}_m^{(b)} = 
        		\mathrm{softmax}\!\left(\frac{\mathbf{Q}_m^{(b)}\mathbf{K}_m^{(b)\top}}{\sqrt{D}}\right),
        	\end{equation}
        where $\mathbf{Q}_m^{(b)}$ and $\mathbf{K}_m^{(b)}$ denote the query and key matrices derived from $\mathbf{A}_m^{(b)}$. The attention score between the \texttt{\texttt{[CLS]}} token and the $n$-th patch token is calculated as
        	\begin{equation}
        		\alpha_{m,n}^{(b)} =
        		\frac{\exp\!\left(\mathbf{q}_{m,0}^{(b)} \!\cdot\! \mathbf{k}_{m,n}^{(b)}\right)}
        		{\sum_{j=1}^{N} \exp\!\left(\mathbf{q}_{m,0}^{(b)} \!\cdot\! \mathbf{k}_{m,j}^{(b)}\right)},
        		\label{alpha_definition}
        	\end{equation}
        where $\mathbf{q}_{m,0}^{(b)}$ and $\mathbf{k}_{m,n}^{(b)}$ denote the query and key vectors corresponding to the \texttt{\texttt{[CLS]}} token and the $n$-th patch token, respectively. The score $\alpha_{m,n}^{(b)}$ quantifies the dependency of the \texttt{\texttt{[CLS]}} token on the $n$-th patch token, serving as an indicator of its semantic importance. Leveraging these scores, each client identifies and retains only the top-$K_m$ most informative tokens (excluding the \texttt{\texttt{[CLS]}} token):
        	\begin{equation}
        		\mathbf{A}_{m}^{\mathrm{sel}} = 
        		\big\{\,\mathbf{a}_{m,n}^{(b)} \,\big|\, n \!\in\! \mathrm{Top}\mathrm{-}K_m(\alpha_{m,1}^{(b)},\ldots,\alpha_{m,N}^{(b)}) \big\}_{b=1}^{B}.
        	\end{equation}
        The number of selected tokens $K_m$ is determined adaptively through the optimization process, thereby reducing redundant tokens and communication overhead while preserving task-relevant information essential for downstream fine-tuning.
	
        \subsubsection{Token Merging} While token selection efficiently filters redundancy, simply discarding low-scoring tokens may lead to the loss of subtle contextual information. To mitigate this, token merging aggregates the discarded tokens into a summary representation. For each image sample $b$ on client $m$, let $\mathcal{S}_m^{(b)}=\mathrm{Top}\text{-}K_m\!\left(\alpha_{m,1}^{(b)},\ldots,\alpha_{m,N}^{(b)}\right)$ and $\mathcal{I}_m^{(b)}=\{1,\ldots,N\}\setminus \mathcal{S}_m^{(b)}$. The discarded tokens are aggregated into a single merged token via an attention-weighted average:
            \begin{equation}
                \mathbf{a}_{m,\mathrm{merge}}^{(b)} =
                \frac{\sum_{n \in \mathcal{I}_m^{(b)}} \alpha_{m,n}^{(b)}\,
                \mathbf{a}_{m,n}^{(b)}}
                {\sum_{n \in \mathcal{I}_m^{(b)}} \alpha_{m,n}^{(b)}}.
            \end{equation}
        Finally, the refined token sequence is constructed by concatenating the \texttt{\texttt{[CLS]}} token, the selected tokens, and the merged token. The overall refined batch representation is expressed as
            \begin{equation}
                \mathbf{A}_m^{\mathrm{ref}} =
                \Big[\,\mathbf{a}_{m,0}^{(b)},\; \{\mathbf{a}_{m,n}^{(b)}\}_{n\in\mathcal{S}_m^{(b)}},\; \mathbf{a}_{m,\mathrm{merge}}^{(b)}\,\Big]_{b=1}^{B},
            \end{equation}
        where $\mathbf{A}_m^{\mathrm{ref}} \in \mathbb{R}^{B \times (K_m + 2) \times D}$. By condensing the discarded tokens rather than removing them, the model preserves global context while substantially reducing the transmission load to the server.
        	\begin{figure}[t]
        		\centering
        		\includegraphics[width=0.8\linewidth]{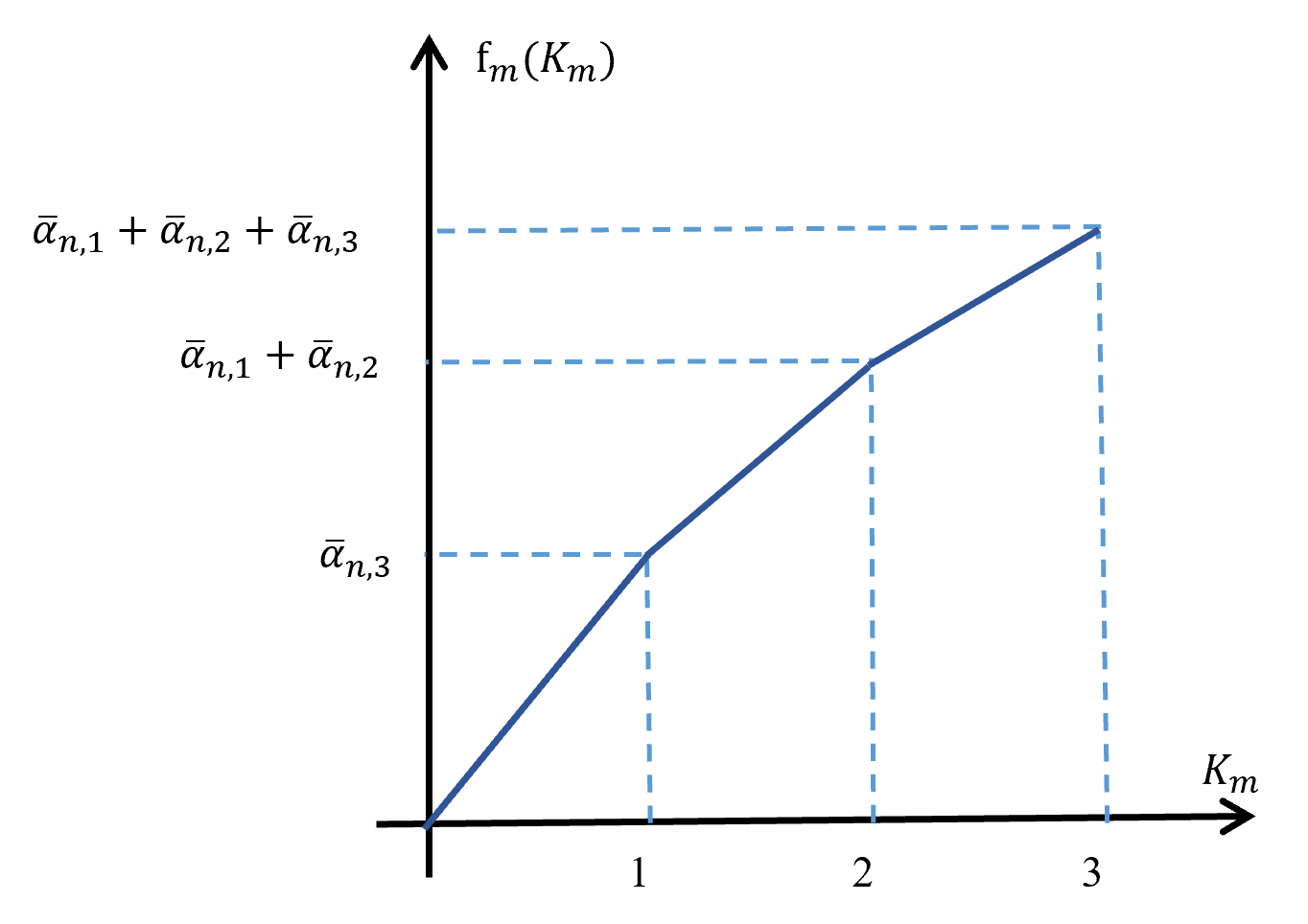}
        		\caption{The cumulative semantic information retention function $f_m(K_m)$.}
        		\label{fig:f_m}
        	\end{figure}

    \section{Problem Formulation}
    \label{section6}
    In this section, we begin by defining a new metric to explicitly quantify the trade-off between semantic information retention and communication latency. Building upon this metric, we then construct the joint optimization problem to determine the optimal token-level transmission and communication resource allocation.

	\subsection{New Metric: Semantic Transmission Efficiency } For each client $m$, the final Transformer block of the client-side model produces a sequence of tokens for each image sample $b$ in the mini-batch:
    	\begin{equation}
    		\mathbf{A}_m^{(b)} = [\mathbf{a}_{m,0}^{(b)}, \mathbf{a}_{m,1}^{(b)}, \ldots, \mathbf{a}_{m,N}^{(b)}] \in \mathbb{R}^{(N+1) \times D},
    	\end{equation}
	where $\mathbf{a}_{m,0}^{(b)}$ denotes the \texttt{\texttt{[CLS]}} token and $\{\mathbf{a}_{m,n}^{(b)}\}_{n=1}^{N}$ are patch tokens. The attention score $\alpha_{m,n}^{(b)}$ measures the task relevance of each token, as defined in (\ref{alpha_definition}).Sorting these scores in descending order yields
    	\begin{equation}
    		\alpha_{m,1}^{(b)} \ge \alpha_{m,2}^{(b)} \ge \cdots \ge \alpha_{m,N}^{(b)}.
    	\end{equation}
	Since resource optimization is performed before each communication round, which transmits a complete mini-batch, a unified token budget $K_m$ must be determined for the entire batch. To maximize the total semantic information retained in each round, we therefore compute a batch-level token importance:
    	\begin{equation}
    		\bar{\alpha}_{m,n} = \sum_{b=1}^{B} \alpha_{m,n}^{(b)}, \quad
    		\bar{\boldsymbol{\alpha}}_m = [\bar{\alpha}_{m,1}, \bar{\alpha}_{m,2}, \ldots, \bar{\alpha}_{m,N}],
    	\end{equation}
	where $\bar{\alpha}_{m,n}$ represents the total semantic contribution of the $n$-th ranked token summed across the mini-batch. The cumulative semantic information retention preserved by selecting the top-$K_m$ tokens is then expressed as:
    	\begin{equation}
    		f_m(K_m) = \sum_{n=1}^{K_m} \bar{\alpha}_{m,n},
    		\label{eq:coverage_single}
    	\end{equation}
	which quantifies the total semantic information retained after filtering. As illustrated in Fig. \ref{fig:f_m}, this function $f_m(K_m)$ exhibits concavity, as formally proven in Lemma \ref{lemma1}.
        \begin{myLem}
            \label{lemma1}
            Let $\{\alpha_{m,n}^{(b)}\}_{n=1}^{N}$ denote the normalized attention scores of the $b$-th sample for client $m$, sorted such that $\alpha_{m,1}^{(b)} \ge \cdots \ge \alpha_{m,N}^{(b)}$. 
            Then, the cumulative semantic information retention function $f_m(K_m)$ defined in (\ref{eq:coverage_single}) is monotonically increasing and concave with respect to $K_m$.
        \end{myLem}
        
        \begin{proof}
            Since $\alpha_{m,n}^{(b)} \ge 0$ and the summed importance scores maintain the non-increasing order $\bar{\alpha}_{m,1} \ge \bar{\alpha}_{m,2} \ge \cdots \ge \bar{\alpha}_{m,N}$, 
            the first-order difference is $f_m(K_m+1)-f_m(K_m)=\bar{\alpha}_{m,K_m+1} \ge 0$, implying monotonicity.
            Furthermore, the inequality $\bar{\alpha}_{m,K_m+1} \le \bar{\alpha}_{m,K_m}$ indicates diminishing marginal gains, thereby establishing the concavity of $f_m(K_m)$ in the discrete domain.
        \end{proof}
	At the system level, the total semantic information retention is derived by summing the individual contributions $\sum_{m \in \mathcal{U}} f_m(K_m)$. To explicitly capture the trade-off between semantic information retention and communication latency, we define the STE as:
    	\begin{equation}
    		\mathcal{E} = \frac{\Delta}{T} = \frac{\sum_{m \in \mathcal{U}} f_m(K_m)}{\underset{m \in \mathcal{U}}{\max}\{T_m^{\mathrm{U}}\}},
    		\label{eq:EITR}
    	\end{equation}
    where the numerator quantifies the total semantic information retained by all participating clients, while the denominator represents the bottleneck communication latency per round, dominated by the straggler due to the synchronous aggregation requirement. As an efficiency metric, $\mathcal{E}$ measures the rate of effective semantic information delivery. A higher value of $\mathcal{E}$ indicates that the system achieves greater learning contribution under lower latency, thereby facilitating more efficient distributed training. As shown in Fig. \ref{fig:ite-tradeoff}, the STE curve exhibits a distinct peak. This validates the existence of an optimal trade-off, where the marginal information gain matches the marginal latency cost.
    	\begin{myRmk}
    		    The metric STE can be interpreted as the \textbf{system-level semantic throughput}. Unlike conventional throughput (measured in bits per second), STE quantifies the density of task-relevant utility delivered per unit of time. Physically, maximizing STE identifies the optimal operating point where the marginal gain of retaining additional tokens matches the marginal cost of the induced time delay from the bottleneck straggler. We use $\max_{m\in\mathcal{U}_t}{T_m^U}$ since STE is a system metric and the uplink phase is bounded by the last-finishing client.
    	\end{myRmk}
    	\begin{figure}[t]
    	    \centering
    	    \includegraphics[width=0.8\linewidth]{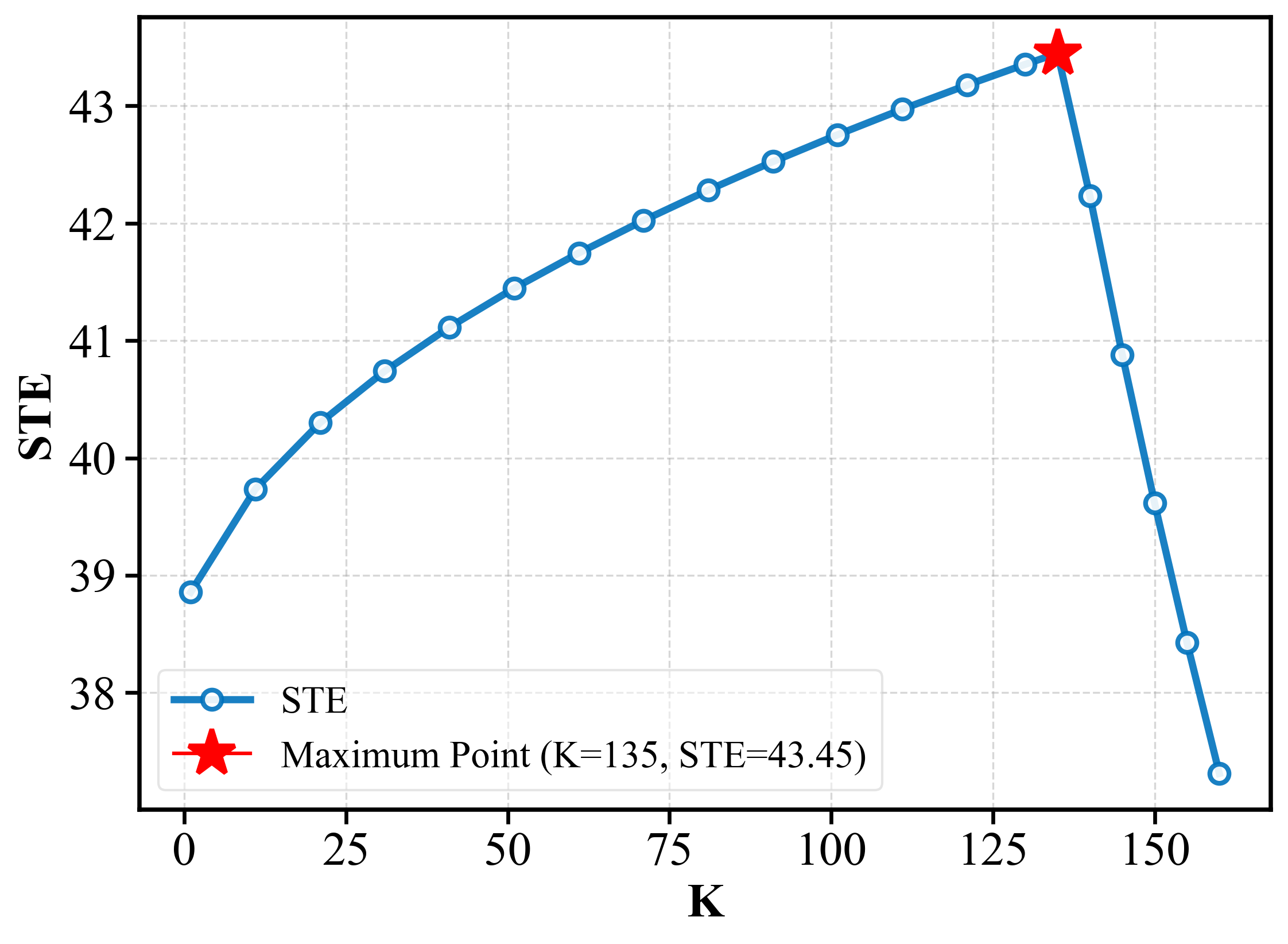}
    	    \caption{STE versus selected token number $K_m$, illustrating the existence of an optimal operating point.}
    	    \label{fig:ite-tradeoff}
    	\end{figure}
    	
	\subsection{Problem Formulation}
	Given the STE metric defined in \eqref{eq:EITR}, we formulate the system-level objective to jointly optimize the token selection $\boldsymbol{K}$, bandwidth allocation $\boldsymbol{W}$, and transmission power $\boldsymbol{p}$. Our goal is to maximize the semantic transmission efficiency in each communication round. This formulation mathematically encapsulates the trade-off between semantic information retention and communication latency:
    	\begin{subequations}
    		\begin{align}
    			\mathcal{P}0:\quad 
    			&\max_{\boldsymbol{K},\,\boldsymbol{W},\,\boldsymbol{p}}
    			\quad
    			\mathcal{E}=\frac{\sum_{m\in\mathcal{U}} f_m(K_m)}
    			{\underset{m\in\mathcal{U}}{\max}\, T_m^{\mathrm U}}
    			\label{eq:system_opt} \\[4pt]
    			\text{s.t.}\quad
    			&\text{C1:}\quad 0 \le p_m \le p^{\max},\quad \forall m\in\mathcal{U}, \label{q0-1}\\[2pt]
    			&\text{C2:}\quad \sum_{m\in\mathcal{U}} W_m \leq W_{\mathrm{tot}}, \label{q0-2}\\[2pt]
    			&\text{C3:}\quad W_m \ge 0,\quad \forall m\in\mathcal{U}, \label{q0-3}\\[2pt]
    			&\text{C4:}\quad K^{\min} \le K_m \le N,\quad K_m\in\mathbb{Z},\quad \forall m\in\mathcal{U}, \label{q0-4}\\[2pt]
    			&\text{C5:}\quad E_m^{\mathrm U}(K_m, W_m, p_m) \le E^{\max},\quad \forall m\in\mathcal{U}, \label{q0-5}\\[2pt]
    			&\text{C6:}\quad T_m^0 + T_m^{\mathrm U}(K_m, W_m, p_m) 
    			\le T_m^{\mathrm{standing}},\quad \forall m\in\mathcal{U}. \label{q0-6}
    		\end{align}
    	\end{subequations}
	Constraint C1 sets the range for transmission power, while Constraints C2 and C3 manage the total bandwidth allocation. These variables are continuous. Constraint C4 restricts the number of selected tokens to be an integer. Finally, Constraints C5 and C6 ensure that the energy consumption and total latency do not exceed the specific limits of each client.

    \section{Problem Solution}
    \label{section7}
	Problem $\mathcal{P}0$ in \eqref{eq:system_opt} jointly optimizes the token selection vector $\boldsymbol{K}$, bandwidth allocation $\boldsymbol{W}$, and transmit power $\boldsymbol{p}$. The objective is a fractional function with a denominator that contains a maximization operator, and the feasible set involves both continuous and integer variables. Hence, $\mathcal{P}0$ is a mixed-integer nonlinear program that is difficult to solve directly. To eliminate the maximization operator in the denominator, we introduce an auxiliary variable
	\begin{equation}
		\tau \triangleq \max_{m \in \mathcal{U}} T_m^{\mathrm{U}}(K_m, W_m, p_m),
	\end{equation}
	which can be enforced through the set of constraints
	\begin{equation}
		T_m^{\mathrm{U}}(K_m, W_m, p_m) \le \tau,\quad \forall m\in\mathcal{U}.
		\label{eq:tau_def}
	\end{equation}
	Using $\tau$, problem $\mathcal{P}0$ can be equivalently rewritten as
	\begin{subequations}
		\label{eq:P0_tau}
		\begin{align}
			\mathcal{P}0':\quad
			&\max_{\boldsymbol{K},\,\boldsymbol{W},\,\boldsymbol{p},\,\tau}
			\ \frac{\sum_{m \in \mathcal{U}} f_m(K_m)}{\tau} \label{eq:P0_tau_obj}\\
			\text{s.t.}\quad
			&\text{C1--C6 in }(\ref{q0-1})\text{--}(\ref{q0-6}), \nonumber\\
			&\text{C7:}\ T_m^{\mathrm{U}}(K_m, W_m, p_m) \le \tau,\ \forall m\in\mathcal{U}. \label{eq:P0_tau_C7}
		\end{align}
	\end{subequations}
	Although the maximization operator is removed, Problem $\mathcal{P}0'$ remains challenging due to the intricate coupling between the integer variable $\boldsymbol{K}$ and the continuous resource variables $(\boldsymbol{W}, \boldsymbol{p})$ in the constraints. To address this, we propose an iterative algorithm. Since $f_m(K_m) \ge 0$ holds for all $m\in\mathcal{U}$, maximizing the fractional objective is then equivalent to minimizing the worst-case latency $\tau$ under the resource, energy, and standing-time constraints. Although this transformation removes the explicit maximization operator from the denominator, problem $\mathcal{P}0'$ remains a nonconvex mixed-integer nonlinear program due to the coupling among $\boldsymbol{K}$, $\boldsymbol{W}$, and $\boldsymbol{p}$. To handle this difficulty, we decompose the original problem into three coupled subproblems, namely power control, bandwidth allocation, and token selection, and solve them iteratively in an alternating manner until convergence.
	\begin{myRmk} 
		Decomposing a highly coupled optimization problem into multiple subproblems and solving them iteratively is commonly adopted when the original function is difficult to optimize directly (see, e.g., \cite{hu2024energy,wu2023split,11045879}). The solutions can be efficiently computed but not necessarily optimal. 
	\end{myRmk}

    \begin{algorithm}[t]
        \caption{Optimal Power via Bisection for $\mathcal{SUBP}1$}
		\label{alg-bisection-new}
        \small
		\begin{algorithmic}[1]
			\REQUIRE Channel gain $\phi_m=\tfrac{h_m}{n_0 W_m}$, parameter $\kappa_m=\tfrac{f_m\ln2}{E^{\max} W_m}$, peak power $p^{\max}$, minimum power $p_m^{\min}$, tolerance $\varepsilon_p>0$.
			\ENSURE Optimal transmit power $p_m^{\star}$ (or infeasibility)
			\STATE \textbf{if} $\dfrac{p^{\max} S_m(K_m)}{W_m\log_2(1+\phi_m p^{\max})} \le E^{\max}$ \textbf{then}
			\STATE \quad \textbf{if} $p^{\max} \ge p_m^{\min}$ \textbf{then} set $p_m^{\star}=p^{\max}$ and \textbf{return};
			\STATE \quad \textbf{else} declare user $m$ infeasible and \textbf{return};
			\STATE \textbf{end if}
			\STATE \textbf{if} $\kappa_m \ge \phi_m$ \textbf{then} declare user $m$ infeasible and \textbf{return};
			\STATE Initialize $\ell=0$, $r=p^{\max}$;
			\REPEAT
			\STATE $p=(\ell+r)/2$;
			\STATE Compute $\Phi=\ln(1+\phi_m p)-\kappa_m p$;
			\IF{$\Phi \ge 0$}
			\STATE $\ell=p$;
			\ELSE
			\STATE $r=p$;
			\ENDIF
			\UNTIL{$(r-\ell)\le \varepsilon_p$};
			\STATE Set $\bar p_m=\ell$ and $p_m^{\mathrm{up}}=\min\{p^{\max},\,\bar p_m\}$;
			\STATE \textbf{if} $p_m^{\min}>p_m^{\mathrm{up}}$ \textbf{then} declare user $m$ infeasible and \textbf{return};
			\STATE \textbf{else} set $p_m^{\star}=p_m^{\mathrm{up}}$ and \textbf{return};
		\end{algorithmic}
	\end{algorithm}

	\subsection{Power Control}
	When only the transmit power vector $\boldsymbol{p}$ is optimized while fixing the token numbers $\boldsymbol{K}$ and bandwidth allocation $\boldsymbol{W}$, the system-level objective in $\mathcal{P}0'$ reduces to minimizing the worst-case uplink latency $\tau = \max_{m\in\mathcal{U}} T_m^{\mathrm{U}}(p_m)$ under the energy and standing-time constraints. Since the transmit powers of different clients are mutually independent, the power control problem can be decomposed into independent subproblems for each device. For a given device $m \in \mathcal{U}$, the power-control subproblem $\mathcal{SUBP}1$ can be expressed as
	\begin{subequations}
		\label{SUBP1-power}
		\begin{align}
			\mathcal{SUBP}1:\quad 
			&\underset{p_m}{\min}\quad 
			T_m^{\mathrm{U}}(p_m) \\
			\text{s.t.}\quad
			&\text{C1}, \text{C5-C6}
		\end{align}
	\end{subequations}
	where $T_m^{\mathrm{U}}(p_m)= \frac{S_m(K_m)}{W_m\log_2\!\left(1+\frac{p_m h_m}{n_0 W_m}\right)}$ denotes the uplink latency of device $m$. For fixed $(K_m,W_m)$, $T_m^{\mathrm{U}}(p_m)$ is strictly decreasing in $p_m$, and minimizing $\tau$ amounts to reducing all $T_m^{\mathrm{U}}$ as much as possible within the feasible power region. From constraint C6, the standing-time constraint of device $m$ can be rewritten as
	\begin{equation}
		T_m^{\mathrm{U}}(p_m) 
		\le T_m^{\max} \triangleq T_m^{\mathrm{standing}} - T_m^0,
	\end{equation}
	which is equivalent to the following lower bound on the transmit power:
	\begin{equation}
		p_m \ge p_m^{\min}
		= \frac{n_0 W_m}{h_m}\!\left(2^{\frac{S_m}{W_m T_m^{\max}}}-1\right).
	\end{equation}
	This specifies the minimum transmit power required to satisfy the latency constraint of device $m$.
	\begin{myTheo}
		$E_m^{\mathrm{U}}(p_m)$ is strictly increasing in $p_m$ for $p_m>0$.
		\begin{proof}
			Let $\phi_m \triangleq \tfrac{h_m}{n_0 W_m}>0$. The uplink communication energy consumption of device $m$ can be rewritten as
			\begin{align}
                \small
				E_m^{\mathrm{U}}(p_m)
				&= \frac{p_m S_m}{W_m\,\log_2(1+\phi_m p_m)} \notag\\
				&= \underbrace{\frac{S_m \ln 2}{W_m}}_{\chi>0}\cdot
				\frac{p_m}{\ln(1+\phi_m p_m)}
				\;\triangleq\; \chi\,\tilde G(p_m),
			\end{align}
			where $\tilde G(p)=\dfrac{p}{\ln(1+\phi_m p)}$. It suffices to show that $\tilde G(p)$ is strictly increasing for $p>0$. The derivative is
			\begin{equation}
                \small
				\tilde G'(p)
				= \frac{\ln(1+\phi_m p)-\frac{\phi_m p}{1+\phi_m p}}
				{\big[\ln(1+\phi_m p)\big]^2}
				= \frac{\psi(\phi_m p)}{\big[\ln(1+\phi_m p)\big]^2},
			\end{equation}
			where $\psi(x)=\ln(1+x)-\tfrac{x}{1+x}$ for $x>0$. Since $\psi'(x)=\tfrac{x}{(1+x)^2}>0$ and $\psi(0)=0$, we have $\psi(x)>0$ for $x>0$. Thus $\tilde G'(p)>0$ for $p>0$, proving that $E_m^{\mathrm{U}}(K_m,W_m,p_m)=\chi\,\tilde G(p_m)$ is strictly increasing on $(0,\infty)$.
			\par\noindent
			This completes the proof.
		\end{proof}
	\end{myTheo}
	Since $T_m^{\mathrm{U}}(p_m)$ is strictly decreasing and $E_m^{\mathrm{U}}(p_m)$ is strictly increasing with respect to $p_m$, the optimal transmit power corresponds to the largest feasible value that satisfies all constraints. To characterize the upper bound induced by the energy constraint C5, define
	\begin{equation}
		\Phi_m(p_m)=\ln(1+\phi_m p_m)-\kappa_m p_m,
		\label{eq:phi_function_new}
	\end{equation}
	where $\kappa_m=\tfrac{S_m\ln2}{E^{\max} W_m},\phi_m=\tfrac{h_m}{n_0 W_m}$. The constraint $E_m^{\mathrm{U}}(p_m)\le E^{\max}$ is equivalent to $\Phi_m(p_m)\ge 0$. We have
	\begin{align}
		\Phi_m'(p_m)&=\frac{\phi_m}{1+\phi_m p_m}-\kappa_m,\qquad\\
		\Phi_m''(p_m)&=-\frac{\phi_m^2}{(1+\phi_m p_m)^2}<0,
	\end{align}
	which implies that $\Phi_m(p_m)$ is strictly concave on $p_m\ge 0$ with $\Phi_m(0)=0$ and $\lim_{p_m\to\infty}\Phi_m(p_m)=-\infty$. Therefore, as illustrated in Fig. \ref{fig:figureshow}, the feasibility structure of the energy constraint can be characterized as follows:
	\begin{itemize}
		\item If  $\frac{p^{\max}S_m}{W_m \log_2\!\left(1+\phi_m p^{\max}\right)} \le E^{\max}$, then $\Phi_m(p^{\max})\ge 0$, indicating that the energy constraint is inactive over $[0,p^{\max}]$ and all powers up to the peak limit are feasible.
		\item If $\kappa_m \ge \phi_m$, then $\Phi_m'(p_m)\le 0$ for all $p_m\ge 0$, which implies 
		$\Phi_m(p_m)\le \Phi_m(0)=0$ for any $p_m>0$.  
		Hence, no strictly positive transmit power can satisfy $E_m^{\mathrm{U}}\!\le E^{\max}$.
		\item If $0<\kappa_m<\phi_m$ and $\Phi_m(p^{\max})<0$, then the strict concavity of $\Phi_m$ guarantees the existence of a unique root $\bar p_m \in (0,p^{\max})$ satisfying $\Phi_m(\bar p_m)=0$.  
		Consequently, the feasible power region induced by the energy constraint is the interval $[0,\,\bar p_m]$.
	\end{itemize}
    
    \begin{figure}[t]
		\centering
		\includegraphics[width=0.9\linewidth]{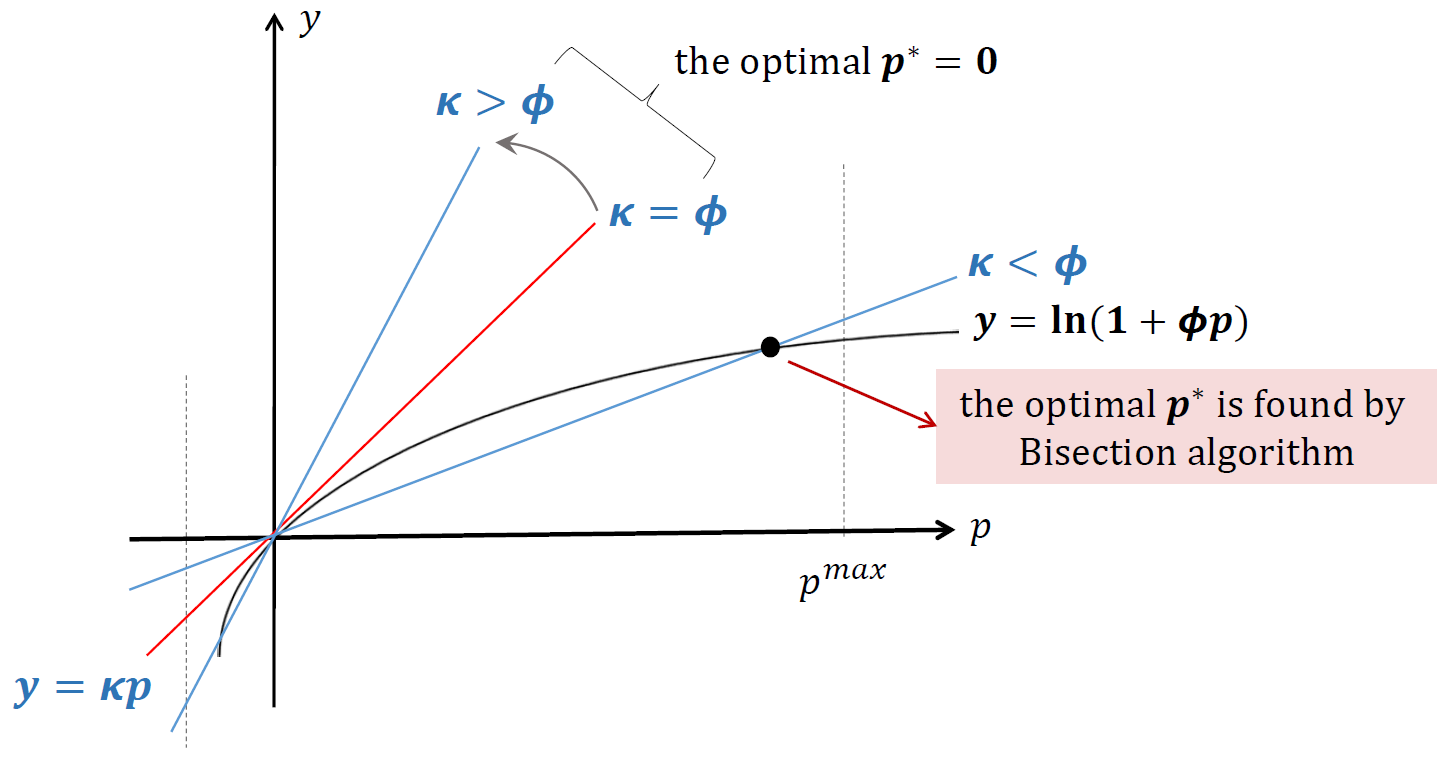}
		\caption{Relationship between $\phi$ and $\kappa$ in the bisection algorithm.}
		\label{fig:figureshow}
	\end{figure}

	Combining the peak-power limit, the latency-induced lower bound, and the energy constraint, the optimal transmit power admits the compact form
	\begin{equation}
		p_m^\star =
		\begin{cases}
			p^{\max}, 
			& E_m^{\mathrm{U}}(p^{\max}) \le E^{\max}
			\ \text{and}\ p^{\max} \ge p_m^{\min}, \\[4pt]
			\min\{p^{\max},\,\bar p_m\},
			& 0<\kappa_m<\phi_m
			\ \text{and}\ p_m^{\min} \le \bar p_m, \\[4pt]
			\text{infeasible},
			& \text{otherwise}.
		\end{cases}
	\end{equation}
	
	For the practical case $0<\kappa_m<\phi_m$ and $E_m^{\mathrm{U}}(p^{\max})>E^{\max}$, the unique root $\bar p_m$ of $\Phi_m(p_m)=0$ can be efficiently found using the \emph{bisection method} over the interval $[0,p^{\max}]$, owing to the strict concavity of $\Phi_m(p_m)$. The details are shown in Alg. \ref{alg-bisection-new}.

    \begin{algorithm}[t]
        \caption{Optimal Bandwidth Allocation for $\mathcal{SUBP}2$}
        \label{alg:subp2-bisection}
        \small
        \begin{algorithmic}[1]
        \REQUIRE Powers $\{p_m\}$, data sizes $\{S_m(K_m)\}$, total bandwidth $W_{\mathrm{tot}}$, tolerance $\varepsilon>0$, and an upper bound $\tau_{\max}$ such that $\Phi(\tau_{\max})\le W_{\mathrm{tot}}$.
        \ENSURE Optimal $(\boldsymbol W^\star,\tau^\star)$.
        
        \STATE Set a small $\tau_{\min}>0$ and initialize $\tau_{\mathrm{low}}=\tau_{\min}$, $\tau_{\mathrm{high}}=\tau_{\max}$.
        \REPEAT
            \STATE $\tau=(\tau_{\mathrm{low}}+\tau_{\mathrm{high}})/2$.
            \FOR{each user $m$}
                \STATE Compute $R_m^{\min}(\tau)$.
                \STATE Compute $W_m^{\min}(\tau)$ via bisection on $W_m\in(0,W_{\mathrm{tot}}]$ solving $R_m^{\mathrm UL}(W_m)=R_m^{\min}(\tau)$.
            \ENDFOR
            \STATE $\Phi(\tau)=\sum_m W_m^{\min}(\tau)$.
            \IF{$\Phi(\tau)>W_{\mathrm{tot}}$}
                \STATE $\tau_{\mathrm{low}}=\tau$.
            \ELSE
                \STATE $\tau_{\mathrm{high}}=\tau$.
            \ENDIF
        \UNTIL{$\tau_{\mathrm{high}}-\tau_{\mathrm{low}}\le\varepsilon$}
        
        \STATE $\tau^\star=\tau_{\mathrm{high}}$ and $W_m^\star=W_m^{\min}(\tau^\star)$.
        \STATE \textbf{return} $(\boldsymbol W^\star,\tau^\star)$.
        \end{algorithmic}
    \end{algorithm}

    \subsection{Bandwidth Allocation}
        Given the transmission power $\boldsymbol p$ and token selection $\boldsymbol K$, the bandwidth allocation subproblem is formulated as:
            \begin{align}
                \min_{\boldsymbol W,\tau}\quad &\tau  \nonumber\\
                \mathrm{s.t.}\quad 
                & C2,\,C3,\,C5,\,C6,\,C7. \nonumber
            \end{align}
        The constraints C2--C7 imposes an effective latency upper bound on each user $m$: $T_m^{\mathrm U}(W_m)\le\bar T_m(\tau)$, where $\bar T_m(\tau)=\min\{\tau,\;E^{\max}/p_m,\;T_m^{\mathrm{standing}}-T_m^0\}$. To satisfy this deadline, the allocated bandwidth must support a minimum transmission rate $R_m^{\mathrm UL}(W_m)\ge R_m^{\min}(\tau)$, with
            \begin{equation}
                \small
                R_m^{\min}(\tau)=\max\left\{\frac{S_m(K_m)}{\tau},\;\frac{p_m S_m(K_m)}{E^{\max}},\;\frac{S_m(K_m)}{T_m^{\mathrm{standing}}-T_m^0}\right\}.
                \label{eq:R_min}
            \end{equation}
        Since the uplink rate function $R_m^{\mathrm UL}(W_m)$ is strictly increasing and concave with respect to $W_m$, it admits a unique inverse $\psi_m(\cdot)$. Consequently, the rate requirement translates to a bandwidth lower bound:
            \begin{equation}
                W_m\ge W_m^{\min}(\tau)\triangleq \psi_m\big(R_m^{\min}(\tau)\big).
            \end{equation}
        Let $\Phi(\tau)=\sum_m W_m^{\min}(\tau)$ denote the total bandwidth required to achieve a system latency of $\tau$. Analyzing \eqref{eq:R_min}, as $\tau$ decreases, the term $S_m(K_m)/\tau$ grows, forcing $R_m^{\min}(\tau)$ and consequently $\Phi(\tau)$ to increase. Specifically, $\Phi(\tau)$ is continuous and strictly decreasing with respect to $\tau$ in the region where $\tau$ is the active bottleneck (i.e., when $S_m/\tau$ determines the maximum in \eqref{eq:R_min}).\par
        For the problem to be feasible, there must exist a $\tau$ such that $\Phi(\tau) \le W_{\text{tot}}$. The optimal latency $\tau^\star$ is the minimum value satisfying this budget constraint. Due to the monotonicity of $\Phi(\tau)$, $\tau^\star$ corresponds to the unique root of the equation:
            \begin{equation}
                \Phi(\tau^\star)=W_{\mathrm{tot}}.
            \end{equation}
        Once $\tau^\star$ is determined, the optimal bandwidth allocation is uniquely given by $W_m^\star=W_m^{\min}(\tau^\star)$. Solving for $(\boldsymbol W^\star,\tau^\star)$ thus reduces to a one-dimensional search for the root of $\Phi(\tau)=W_{\mathrm{tot}}$. Since the inverse Shannon capacity $\psi_m(\cdot)$ lacks a closed-form expression in terms of elementary functions, we calculate $W_m^{\min}(\tau)$ numerically.  The overall procedure employs a nested bisection method: the outer loop searches for $\tau^\star$, while the inner loop inverts the rate function, as detailed in Algorithm~\ref{alg:subp2-bisection}.

	\begin{algorithm}[t]
		\caption{Joint Optimization of Power, Bandwidth, and Token Selection}
		\label{alg:joint}
        \small
		\begin{algorithmic}[1]
			\REQUIRE Initial feasible point $(\mathbf{p}^0,\mathbf{W}^0,\mathbf{K}^0,\tau^0)$, tolerances $\varepsilon_p,\varepsilon_W,\varepsilon_K,\varepsilon_\tau>0$. Set $i=0$.
			\REPEAT
			\STATE Compute the optimal transmit power vector $\mathbf{p}^{\,i+1}$ by solving $\mathcal{SUBP}1$ for all users given $(\mathbf{W}^i,\mathbf{K}^i)$ using Algorithm~\ref{alg-bisection-new}.
			\STATE Compute the optimal bandwidth allocation $\mathbf{W}^{\,i+1}$ and latency bound $\tau^{\,i+1}$ by solving $\mathcal{SUBP}2$ given $(\mathbf{p}^{\,i+1},\mathbf{K}^i)$ using Algorithm\ref{alg:subp2-bisection}.
			\STATE Compute the optimal token numbers $\mathbf{K}^{\,i+1}$ by solving $\mathcal{SUBP}3$ in closed form given $(\mathbf{p}^{\,i+1},\mathbf{W}^{\,i+1},\tau^{\,i+1})$.
			\STATE $i \leftarrow i+1$;
			\UNTIL{
				$\|\mathbf{p}^{\,i}-\mathbf{p}^{\,i-1}\|<\varepsilon_p$
				\ \textbf{and} \
				$\|\mathbf{W}^{\,i}-\mathbf{W}^{\,i-1}\|<\varepsilon_W$
				\ \textbf{and} \
				$\|\mathbf{K}^{\,i}-\mathbf{K}^{\,i-1}\|<\varepsilon_K$
				\ \textbf{and} \
				$|\tau^{\,i}-\tau^{\,i-1}|<\varepsilon_\tau$
			}
			\ENSURE $\mathbf{p}^{\,i},\mathbf{W}^{\,i},\mathbf{K}^{\,i},\tau^{\,i}$.
		\end{algorithmic}
	\end{algorithm}

    \subsection{Token Selection}
	    Given the bandwidth allocation $\boldsymbol{W}$, transmit power $\boldsymbol{p}$, and latency bound $\tau$, the token-selection subproblem of $\mathcal{P}0'$ becomes
        	\begin{subequations}
        		\label{eq:SUBP-token}
        		\begin{align}
        			\mathcal{SUBP}3:\quad
        			&\max_{\boldsymbol{K}} \ \sum_{m\in\mathcal{U}} f_m(K_m) \nonumber \\
        			\text{s.t.}\quad \text{C4}-\text{C7}.\nonumber
        		\end{align}
        	\end{subequations}
	The payload size grows linearly with the number of selected tokens, as $S_m(K_m)=B\,(K_m+2)\,D\,q_0$. Since $(W_m,p_m)$ are fixed in this subproblem, the uplink rate $R_m^{\mathrm{UL}}=W_m\log_2\!\left(1+\frac{p_m h_m}{n_0 W_m}\right)$ is a constant, and the uplink latency and energy consumption reduce to	$T_m^{\mathrm{U}}=\frac{\beta_m (K_m+2)}{R_m^{\mathrm{UL}}}, E_m^{\mathrm{U}}(K_m) = p_m\,T_m^{\mathrm{U}}(K_m)$,where $\beta_m = BDq_0 $ denotes the number of transmitted bits per token. Substituting these expressions into the constraints yields the following linear upper bounds on $K_m$:
    	\begin{align}
    		&K_m \le \frac{E^{\max}R_m^{\mathrm{UL}}}{p_m\beta_m} - 2,\\
    		&K_m \le \frac{(T_m^{\mathrm{standing}}-T_0)R_m^{\mathrm{UL}}}{\beta_m} - 2,\\
    		&K_m \le \frac{\tau R_m^{\mathrm{UL}}}{\beta_m} - 2.
    	\end{align}
	Hence, the maximum feasible token number for client $m$ is
        \begin{equation}
            \begin{aligned}
            K_m^{\max} &=
            \left\lfloor
            \min\!\left\{
            N,\ 
            \frac{E^{\max}R_m^{\mathrm{UL}}}{p_m\beta_m}-2,\right.\right.\\
            &\qquad\left.\left.
            \frac{(T_m^{\mathrm{standing}}-T_m^0)R_m^{\mathrm{UL}}}{\beta_m}-2,\ 
            \frac{\tau R_m^{\mathrm{UL}}}{\beta_m}-2
            \right\}
            \right\rfloor .
            \end{aligned}
        \end{equation}

	If $K_m^{\max} < K^{\min}$, client $m$ is infeasible under the current $(W_m,p_m,\tau)$.
	Otherwise, the per-client decision reduces to the one-dimensional discrete maximization
	\begin{equation}
		K_m^\star=\arg\max_{K^{\min}\le K_m\le K_m^{\max}} f_m(K_m).
	\end{equation}
    Since $f_m(K_m)$ is monotonically increasing in $K_m$, the objective is maximized by choosing the largest feasible token budget. Hence, the optimal solution to this subproblem is achieved at
        \begin{equation}
            K_m^\star = K_m^{\max}.
        \end{equation}
    Accordingly, each client selects the maximum number of informative tokens permitted by its latency and energy constraints.

    \subsection{Iterative Joint Optimization and Complexity} The original problem is decomposed into three subproblems and solved alternately, as shown in Alg.~\ref{alg:joint}. In $\mathcal{SUBP}1$, each user computes its optimal transmit power via a one-dimensional bisection search, yielding a per-iteration complexity of $\mathcal{O}\!\left(U\log(1/\varepsilon_p)\right)$. In $\mathcal{SUBP}2$, $\tau^\star$ is obtained by an outer bisection search, and for each candidate $\tau$ the minimum required bandwidth $W_m^{\min}(\tau)$ is computed via an inner bisection to invert $R_m^{\mathrm{UL}}(W_m)$. This results in $\mathcal{O}\!\left(U\log(1/\varepsilon_W)\log(1/\varepsilon_\tau)\right)$ complexity, where $\varepsilon_\tau$ and $\varepsilon_W$ are the tolerances of the outer and inner bisection loops, respectively. In $\mathcal{SUBP}3$, each user obtains the token number by evaluating the maximum feasible $K_m$, which incurs $\mathcal{O}(U)$ complexity. Therefore, the the total complexity is $\mathcal{O}\!\left(I_{\mathrm{iter}}\!\left[U\log(1/\varepsilon_p) + U\log(1/\varepsilon_W)\log(1/\varepsilon_\tau)\right]\right)$, where $I_{\mathrm{iter}}$ is the number of outer iterations required for convergence.

\begin{table*}[t]
    \caption{Top-1 Accuracy (\%) on ImageNet100, Oxford Flowers-102, and CUB-200-2011.}
    \centering 
    \scriptsize
    \setlength{\tabcolsep}{7pt}
    \renewcommand{\arraystretch}{1.1} 
    \begin{tabular}{c|c|cc|cc|cc}
    \hline
     & & \multicolumn{6}{c}{\textbf{Datasets}}\\ \cline{3-8} 
    \textbf{Backbone} & \textbf{Method} & \multicolumn{2}{c}{\textbf{ImageNet100}}  & \multicolumn{2}{c}{\textbf{Oxford Flowers-102}} & \multicolumn{2}{c}{\textbf{CUB-200-2011}} \\
    \cline{3-8}
     &  & \textit{IID} & \textit{Non-IID} & \textit{IID} & \textit{Non-IID} & \textit{IID} & \textit{Non-IID} \\
    \hline
    \multirow{6}{*}{ViT-S/16}
    & LocalLoRA  & 68.42 & 37.04 & 84.11 & 85.09 & 42.59 & 44.21 \\
    & FedLoRA    & 64.23 & 39.69 & 74.21 & 55.75 & 41.33 & 17.33 \\
    & SplitLoRA  & 85.16 & 84.34 & 98.00 & 99.61 & 82.21 & 75.91 \\
    & SFLora     & 84.20 & 83.24 & 98.85 & 99.19 & 81.62 & 74.49 \\
    & ST-SFLora-Full & 84.36 & 83.13 & 98.83 & 99.10 & 81.39 & 74.53 \\
    & ST-SFLora (Ours) & 80.77 & 79.47 & 97.26 & 97.67 & 78.21 & 68.86 \\
    \hline
    \multirow{6}{*}{ViT-B/16}
    & LocalLoRA & 80.78 & 45.64 & 80.44 & 75.43 & 37.09 & 38.97 \\
    & FedLoRA   & 81.56 & 52.49 & 74.11 & 57.95 & 59.83 & 21.53 \\
    & SplitLoRA & 90.08 & 89.47 & 99.25 & 99.29 & 84.79 & 80.65 \\
    & SFLora    & 89.55 & 89.09 & 98.97 & 98.69 & 84.55 & 79.84 \\
    & ST-SFLora-Full & 89.48 & 89.14 & 99.00 & 98.74 & 84.57 & 79.91 \\
    & ST-SFLora (Ours) & 87.10 & 85.81 & 96.72 & 96.43 & 80.02 & 73.69 \\
    \hline
    \multirow{6}{*}{ViT-L/16}
    & LocalLoRA   & 82.52 & 44.86  & 93.89 & 92.54 & 58.06 & 57.76 \\
    & FedLoRA     & 83.02 & 51.32  & 93.86 & 73.72 & 67.15 & 32.56 \\
    & SplitLoRA   & 90.47 & 89.79  & 99.61 & 99.73 & 87.96 & 84.92 \\
    & SFLora      & 90.39 & 89.50  & 99.61 & 99.73 & 87.70 & 84.39 \\
    & ST-SFLora-Full & 90.49 & 89.51 & 99.61 & 99.73 & 87.75 & 84.11 \\
    & ST-SFLora (Ours) & 87.20 & 85.45 & 99.00 & 99.14 & 81.36 & 75.83 \\
    \hline
    \end{tabular}
    \label{tab:big_comparison}
\end{table*}     		
\begin{table}[t]
    \caption{Client-side model computation and communication overhead.}
    \centering
    \scriptsize
    \setlength{\tabcolsep}{4pt}
    \newcolumntype{C}[1]{>{\centering\arraybackslash}m{#1}}
    \begin{tabular}{C{2.0cm} C{2.0cm} C{0.8cm} C{0.8cm} C{0.8cm}}
    \toprule
    \multirow{2}{*}{\textbf{Method}} &
    \textbf{Computation (GB)} &
    \multicolumn{3}{c}{\textbf{Communication (MB)}} \\
    \cmidrule(lr){2-2} \cmidrule(lr){3-5}
    & \textbf{GPU Mem.} &
      \textbf{Model} &
      \textbf{LoRA} &
      \textbf{Token} \\
    \midrule
    LocalLoRA           & 9.0 & 335.3 & 7.9 & 0 \\
    FedLoRA             & 9.0 & 335.3 & 7.9 & 0 \\
    SplitLoRA           & 2.3 & 58.2  & 1.3 & $\tfrac{3}{16}N$ \\
    SFLora              & 2.3 & 58.2  & 1.3 & $\tfrac{3}{16}N$ \\
    ST-SFLora-Full      & 1.4 & 0     & 1.3 & $\tfrac{3}{16}N$ \\
    ST-SFLora (top-$K$) & 1.4 & 0     & 1.3 & $\tfrac{3}{16}(K+2)$ \\
    \bottomrule
    \end{tabular}
    \label{tab:client_resource_cost}
\end{table}
        
	\section{Performance Evaluation}
        \label{section8}
    	\subsection{Experimental Setting}  
            We consider a distributed wireless edge network consisting of one central server and $100$ mobile clients. To simulate dynamic device availability, the number of participating clients in each round is modeled by a Poisson distribution. These active clients are then considered for the joint optimization and training process. The clients are uniformly distributed within a coverage area with a radius of $5$ to $500$m, communicating with the server via wireless uplink channels.  The maximum transmit power for each client is set to $0.2$~W, and the total system uplink bandwidth is $50$~MHz. The maximum GPU frequencies for the devices are selected from the intervals $[1.0, 1.5]$ GHz. The channel gain accounts for large-scale fading following a path-loss model with an exponent of $2.5$. The noise power spectral density is set to $-174$~dBm/Hz. We assume that the mobile devices are equipped with GPUs, such as the Apple A15, which features 4 to 6 cores. Accordingly, the number of GPU cores on the device side, denoted as $C_n$, is selected from the range $[4, 6]$. The computational capability is normalized such that the FLOPs per cycle per core is set to $D_n=1$ for all clients.\par
            We evaluate the proposed ST-SFLora framework on distributed image classification tasks employing Vision Transformer (ViT) backbones. All models are initialized with pretrained weights obtained from the \textit{timm} library. Unless stated otherwise, the default hyperparameters are set as follows: learning rate $0.01$, batch size $64$, and LoRA rank $16$. Experiments are conducted on three benchmarks covering both generic and fine-grained classification tasks: \textbf{ImageNet100}~\cite{russakovsky2015imagenet}, a widely-used subset of ILSVRC-2012 containing 100 classes, serving as a standard benchmark for generic object recognition; \textbf{Oxford Flowers-102}~\cite{nilsback2008automated}, a fine-grained dataset consisting of 102 flower categories, characterized by large scale, pose, and light variations within classes; and \textbf{CUB-200-2011}~\cite{wah2011caltech}, a challenging fine-grained dataset comprising 200 bird species, which requires capturing subtle discriminative features to distinguish between visually similar sub-categories. We consider both IID and non-IID data partition settings. For the non-IID scenarios, the local class distributions are generated using a Dirichlet distribution with a concentration parameter $\alpha=0.5$ to simulate data heterogeneity.

            \begin{figure*}[ht]
                \centering
                \subfloat[Convergence of the joint optimization algorithm.]{
                    \includegraphics[width=0.33\linewidth]
                    {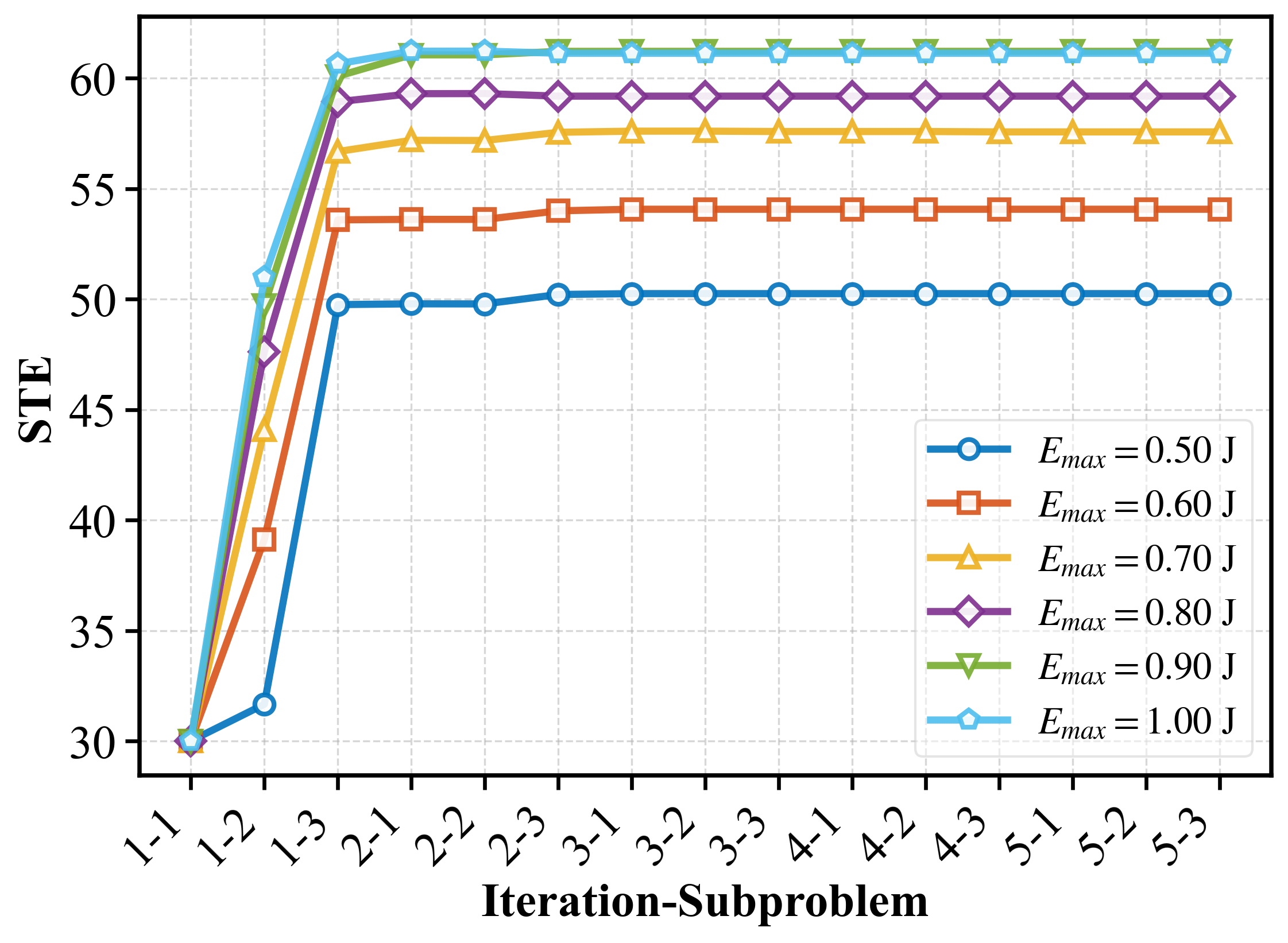}
                    \label{a}
                }\hfill
                \subfloat[Comparison of STE across algorithms.]{
                    \includegraphics[width=0.31\linewidth]
                    {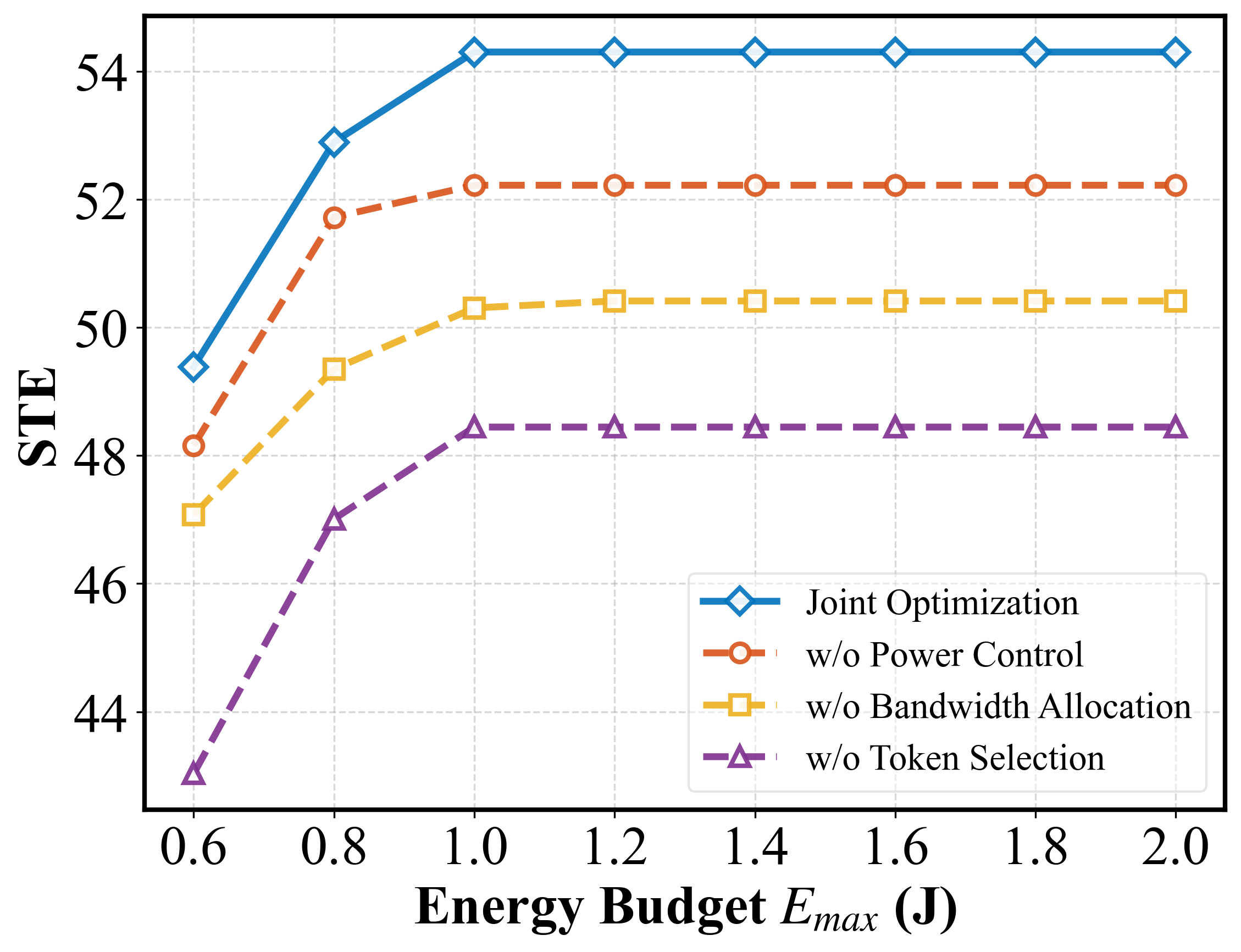}
                    \label{b}
                }\hfill
                \subfloat[The average token amount of varying $W_{\mathrm{tot}}$ and $E_{\mathrm{max}}$.]{
                    \includegraphics[width=0.33\linewidth]
                    {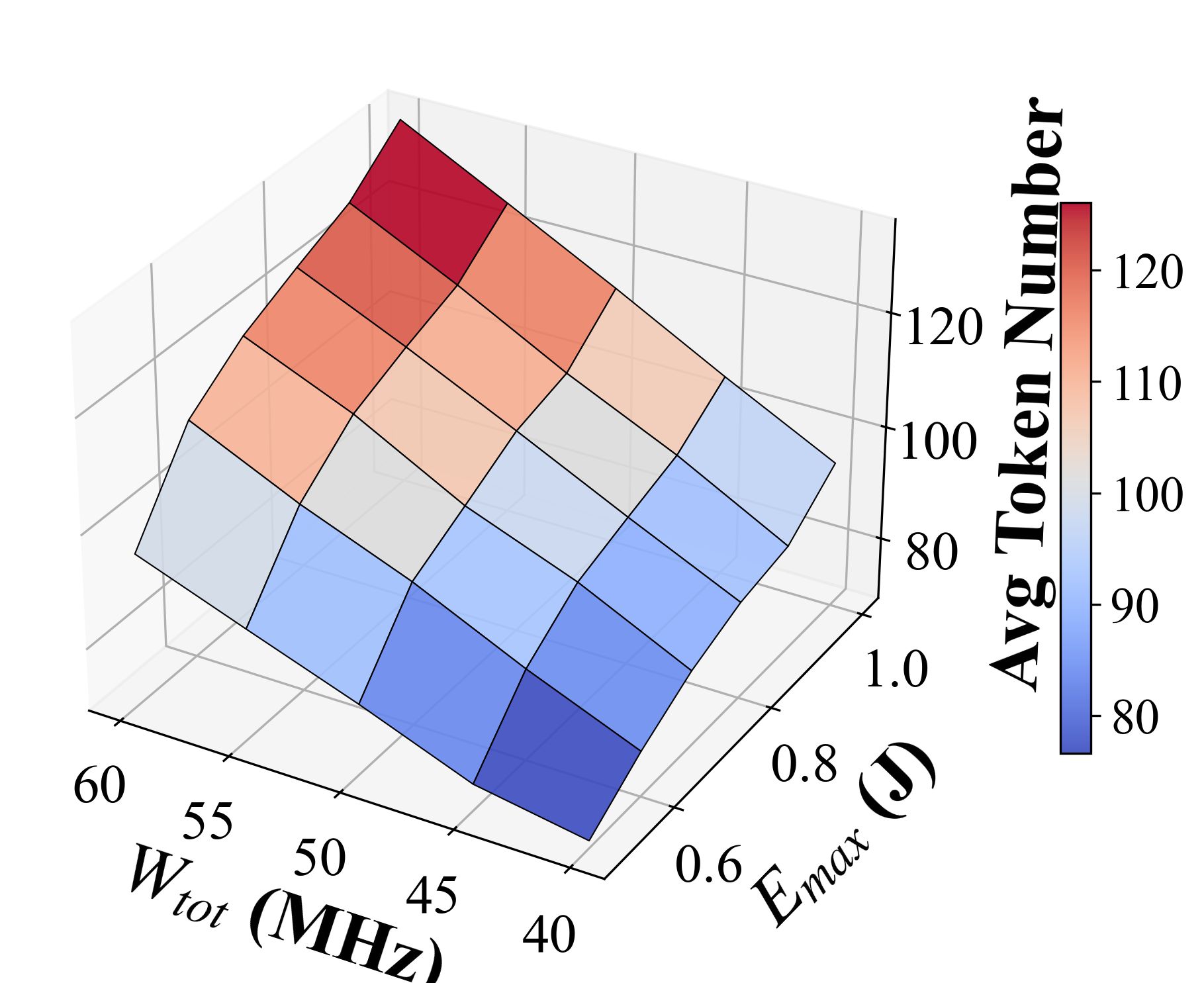} 
                    \label{c}
                }
                \caption{Optimization Algorithm Performance}
                \label{fig:optimization-results}
            \end{figure*}
            \begin{figure}[t]
                \centering
                \includegraphics[width=1.0\linewidth]{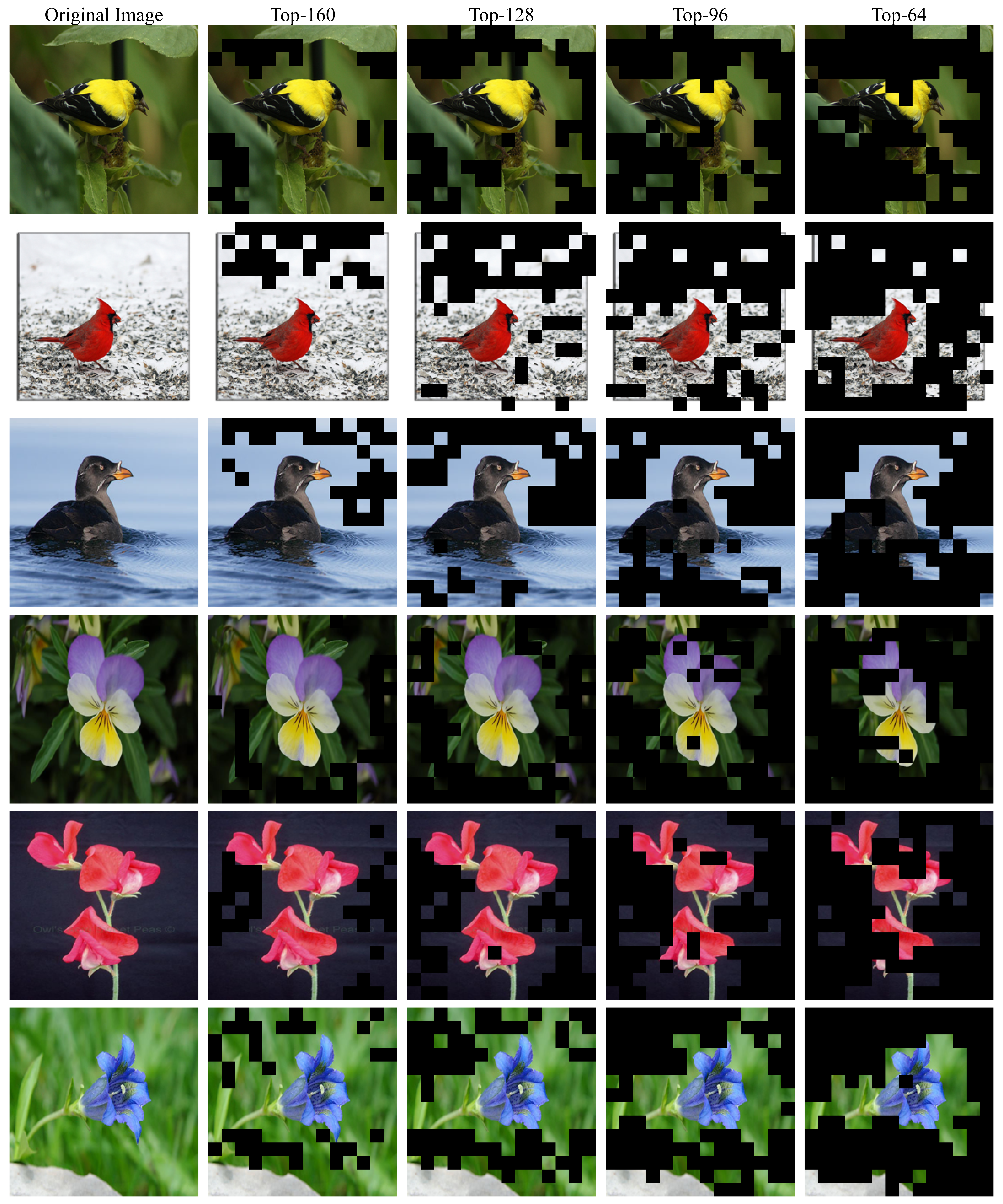}
                \caption{Visualization of proposed tokens selection strategy.}
                \label{fig:token-vis}
            \end{figure}
        
        \subsection{Performance under Different Distributed Architectures}
            Table~\ref{tab:big_comparison} compares six representative distributed fine-tuning architectures. LocalLoRA performs LoRA-based fine-tuning locally without communication. FedLoRA follows the standard FL pipeline in which clients fine-tune LoRA modules locally and upload them for aggregation. SplitLoRA adopts the SL paradigm, where devices sequentially collaborate with the server to train a partitioned model with LoRA modules on both sides. SFLora extends LoRA into the SFL architecture, enabling parallel client participation while retaining trainable modules on both sides. ST-SFLora-Full is our proposed SFL variant that freezes all client-side parameters without communication token selection. ST-SFLora further incorporates the proposed semantic token selection mechanism to reduce uplink traffic.\par
            LocalLoRA and FedLoRA exhibit the lowest accuracy, particularly under non-IID settings. LocalLoRA suffers from limited data diversity and insufficient data volume. In contrast, FedLoRA aggregates only the LoRA updates, whose parameter space is substantially smaller than the full model. Under heterogeneous data, these low-rank updates often point in inconsistent task-adaptive directions, leading to mutual cancellation during FedAvg aggregation. Consequently, both methods remain highly sensitive to client drift and show pronounced performance degradation under non-IID distributions.\par
            SplitLoRA often achieves the strongest performance, and data heterogeneity has only a marginal impact on the SL framework. SFLora attains accuracy close to SplitLoRA, with differences typically within around $2\%$. ST-SFLora-Full achieves accuracy comparable to SplitLoRA across datasets; on CUB-200-2011 with ViT-Large, it attains 87.75\%, nearly matching the 87.96\% of SplitLoRA. ST-SFLora (top-$K$), which transmits only semantically important tokens, incurs moderate accuracy reductions but consistently outperforms all FL-based baselines, even under aggressive compression.\par
            Table~\ref{tab:client_resource_cost} further shows that split-based designs substantially reduce client-side resource usage when adapting ViT-B/16. The GPU memory requirement decreases from 9.0\,GB (LocalLoRA/FedLoRA) to 2.3\,GB (SplitLoRA/SFLora), and further to 1.4\,GB with ST-SFLora. In addition, LocalLoRA and FedLoRA require broadcasting the entire model at the beginning of training, exceeding 335\,MB of parameters, whereas all split-based variants eliminate this overhead entirely. Their communication cost is dominated by activation transmission, which reduces from $\tfrac{3N}{16}$\,MB to $\tfrac{3K}{16}$\,MB under top-$K$ token selection. Here, $\tfrac{3}{16}$\,MB corresponds to the activation footprint of a single token for one batch, computed as $\frac{3}{16} = \frac{B \cdot D \cdot q_0}{1024^2} = \frac{64 \times 768 \times 32}{1024^2}~\text{MB}.$ These results indicate that ST-SFLora incurs the lowest computational and communication overheads among all compared methods. While this aggressive resource reduction leads to a slight decrease in accuracy, the performance degradation remains within a tolerable range, thereby ensuring the feasibility of fine-tuning on strictly constrained mobile devices.

        \subsection{Impact of Resource Allocation Optimization}
            The impact of the proposed joint resource allocation strategy is evaluated in Fig. \ref{fig:optimization-results}. The optimization framework integrates power control, bandwidth allocation, and token selection, and the results demonstrate their synergistic contribution to maximizing the system's STE. Fig.~\ref{fig:optimization-results}\subref{a} illustrates the convergence of the iterative joint optimization algorithm under different energy budgets $E_{\max}$. In all cases, the objective value (STE) rises rapidly within the first few iterations and stabilizes within fewer than five outer loops. A larger $E_{\max}$ yields a higher converged STE, as increased energy availability relaxes transmission constraints, enabling higher transmit power, wider bandwidth utilization, or a larger token budget. These results verify the computational efficiency and stability of the proposed optimization procedure.\par
            Fig.~\ref{fig:optimization-results}\subref{b} compares the performance of the full joint optimization against three ablated variants. Removing power control yields a moderate performance drop, indicating that adapting transmit power to channel heterogeneity is beneficial. The absence of bandwidth allocation causes a larger degradation, as fixed bandwidth division limits the achievable uplink rates. The most significant drop occurs when token selection is disabled. This result highlights that transmitting all tokens induces excessive latency overheads that outweigh the marginal semantic gain, thereby drastically reducing the transmission efficiency. The full configuration consistently achieves the highest STE across all energy budgets, demonstrating that all three components are essential for system performance.\par
            Fig.~\ref{fig:optimization-results}\subref{c} presents the average number of selected tokens under varying bandwidth $W_{\mathrm{tot}}$ and energy budgets $E_{\max}$. Both resources are positively correlated with the token count: increasing either dimension allows the system to preserve more semantic information. The observed surface trend indicates that energy and bandwidth act as coupled bottlenecks: limiting either resource restricts the allowable token budget, and the algorithm adaptively reduces token usage when resources become scarce. This observation highlights the complementary roles of physical-layer and semantic-layer resource management. 

                \begin{figure*}[t]
                    \centering
                    \subfloat[ImagNet100]{\includegraphics[width=0.32\linewidth]{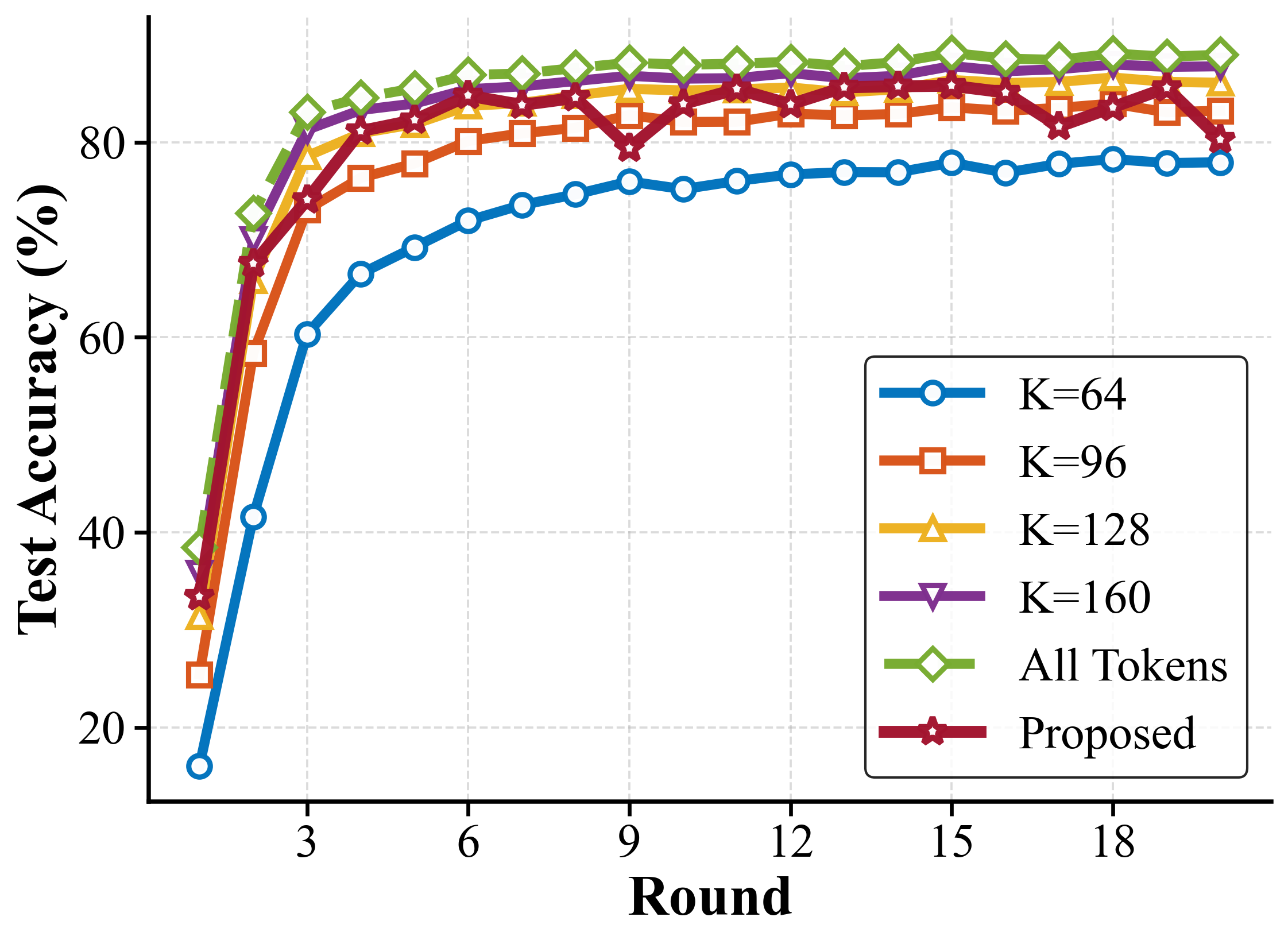}}\hfill
                    \subfloat[Oxford Flowers-102]{\includegraphics[width=0.32\linewidth]{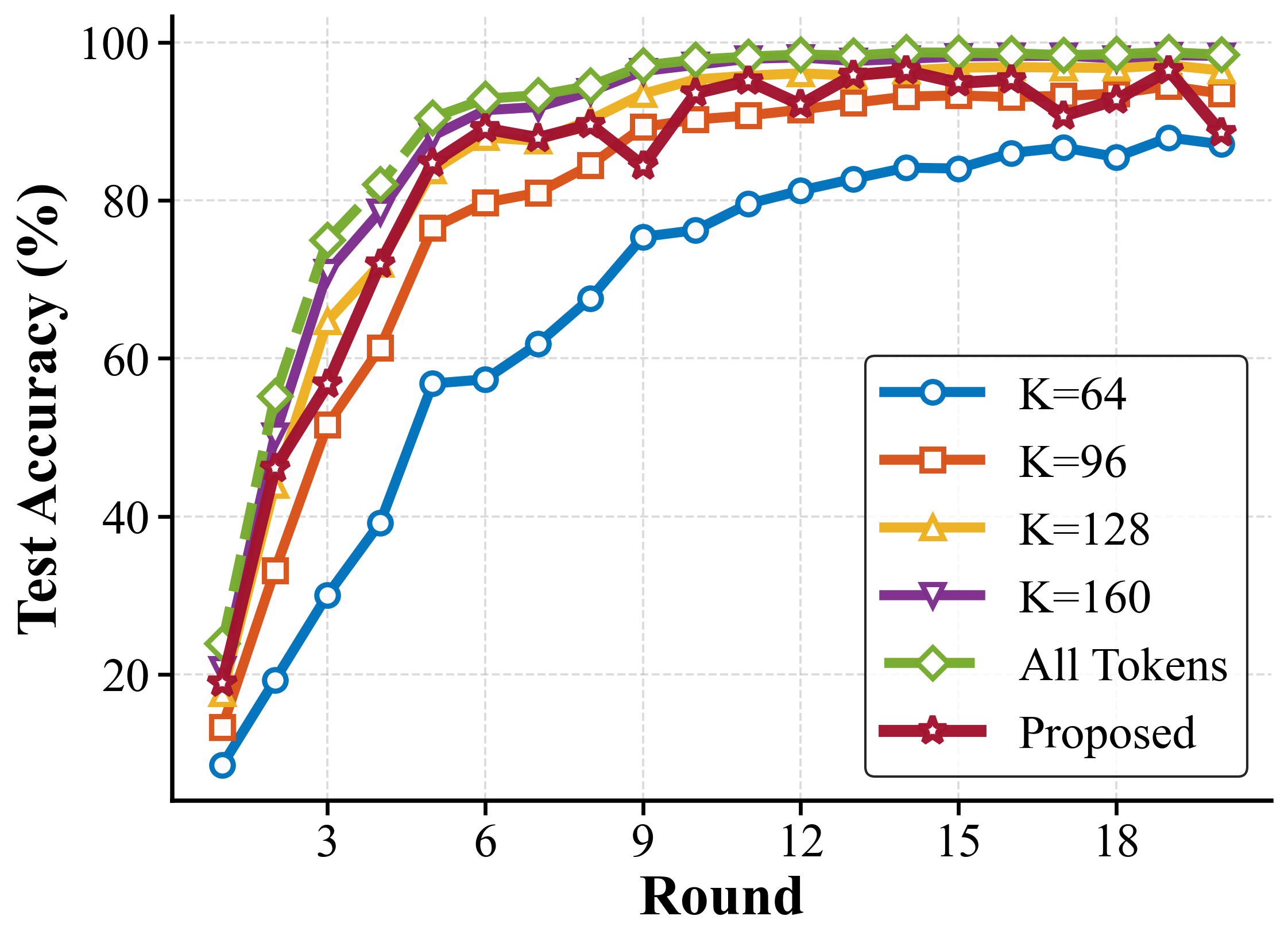}}\hfill
                    \subfloat[CUB-200-2011]{\includegraphics[width=0.32\linewidth]{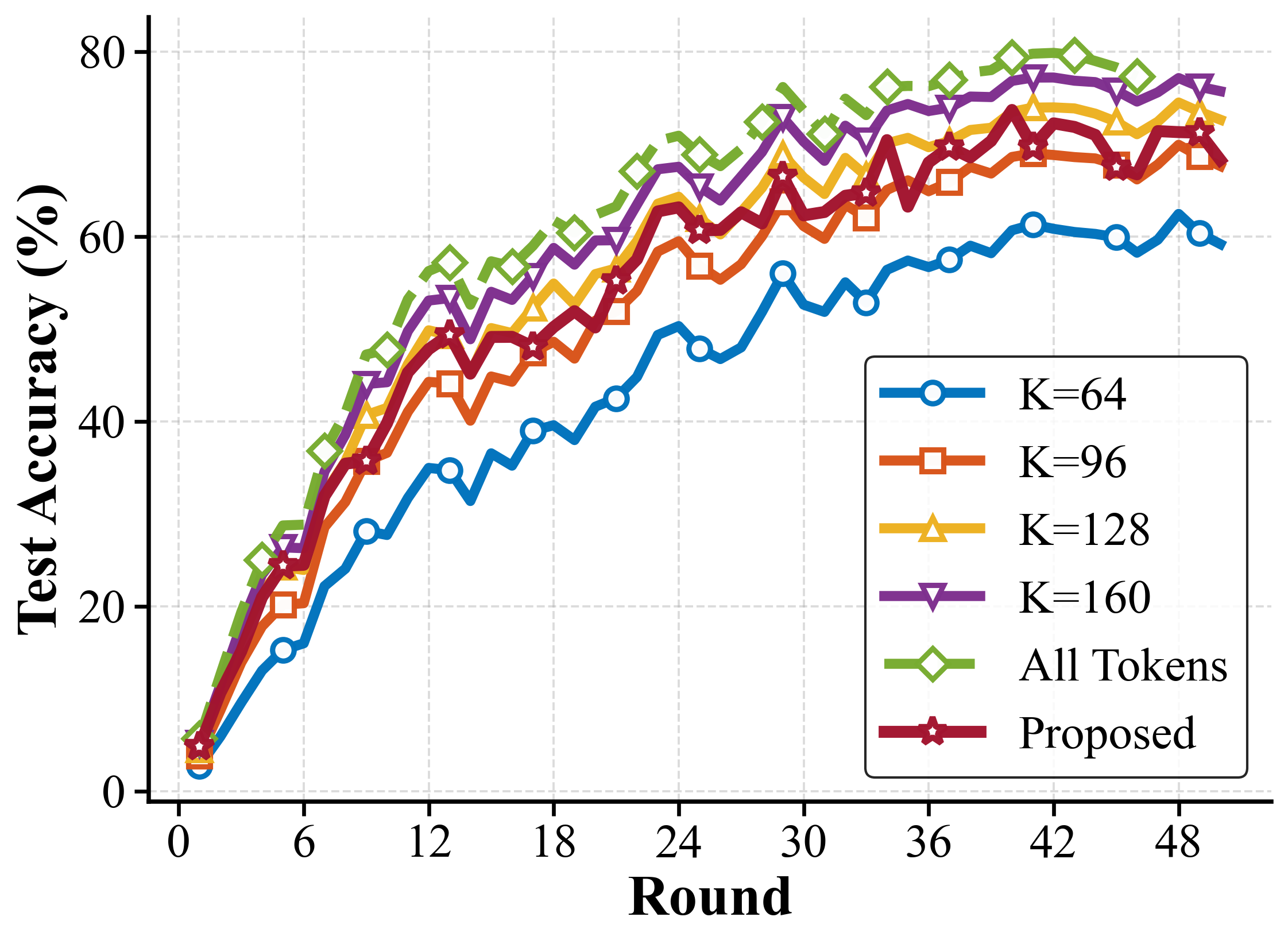}}
                    \caption{Testing Accuracy on different datasets with varing selected token
                     number.}
                    \label{fig:test-compare-token-number}
                \end{figure*}
    
        \subsection{Effectiveness of Token Selection Strategy}
            The effectiveness of the proposed token selection strategy is validated through both visualization and quantitative evaluations, as shown in Figs.~\ref{fig:token-vis} and~\ref{fig:test-compare-token-number}. Fig.~\ref{fig:token-vis} presents representative examples of the selected tokens under varying budgets $K$. As $K$ decreases from $160$ to $64$, the model retains progressively fewer tokens. The masked regions (shown in black) primarily correspond to background clutter, while the preserved tokens concentrate on regions containing class-discriminative attributes, such as object contours, textures, and salient color patterns. This trend is consistent across bird, flower, and generic object categories, indicating that the proposed metric effectively prioritizes semantically informative patches and suppresses visual redundancy.\par
            Fig.~\ref{fig:test-compare-token-number} reports the test accuracy for different token budgets ($K \in \{64, 96, 128, 160\}$) compared to the full-token baseline across three datasets. Several key observations can be made. First, increasing the token budget generally improves the final accuracy. Across ImageNet100, Flowers-102, and CUB-200-2011, the performance gap between $K = 64$ and $K = 160$ is noticeable, particularly in the fine-grained classification tasks. A larger token budget preserves more spatial detail, leading to higher final accuracy. Second, for $K \ge 128$, the accuracy curves closely approximate the all-token baseline, suggesting that a moderate token budget is sufficient to encapsulate task-relevant semantics. In ImageNet100 and Flowers-102, $K = 160$ achieves performance nearly identical to the full-token model throughout the training process. In the challenging CUB dataset, although the all-token input remains slightly superior, $K = 160$ and $K = 128$ still achieve competitive accuracy with significantly reduced communication overhead. Third, ST-SFLora exhibits accuracy fluctuations reflecting its constraint-aware nature. To handle time-varying channels, the algorithm proactively downscales the token budget during deep fading. This mechanism trades marginal short-term accuracy for ensured system feasibility under strict resource constraints, demonstrating a robustness trade-off that static baselines lack. 
    \section{Conclusion}
        \label{section9}
        In this paper, we proposed ST-SFLora, a semantic token-based Split Federated LoRA fine-tuning framework designed to enable efficient computation and communication of large-scale Transformer-based vision models on resource-constrained mobile edge devices. We defined a Semantic Transmission Efficiency (STE) metric to quantify the system trade-off between semantic information retention and communication costs. To further enhance system efficiency, we formulated a joint optimization problem that coordinates transmit power, bandwidth allocation, and token selection to maximize the system STE. Extensive experiments on various datasets and backbones demonstrate that ST-SFLora achieves the lowest client-side resource consumption among state-of-the-art baselines.\par
        Future work could explore more effective, learnable token selection strategies to further refine the efficiency trade-off. Additionally, generalizing the framework to other domains such as Natural Language Processing (NLP) would broaden the scope of efficient fine-tuning for foundation models.
        
	\bibliographystyle{IEEEtran}
	\bibliography{main.bib}

@inproceedings{devlin2019bert,
    title={Bert: Pre-training of deep bidirectional transformers for language understanding},
    author={Devlin, Jacob and Chang, Ming-Wei and Lee, Kenton and Toutanova, Kristina},
    booktitle={Proceedings of the 2019 conference of the North American chapter of the association for computational linguistics: human language technologies, volume 1 (long and short papers)},
    pages={4171--4186},
    year={2019}
}

@article{dosovitskiy2020image,
    title={An image is worth 16x16 words: Transformers for image recognition at scale},
    author={Dosovitskiy, Alexey},
    note ={\textit{arXiv:2010.11929}},
    year={2020}
}

@misc{achiam2023gpt,
    title={Gpt-4 technical report},
    author={Achiam, Josh and Adler, Steven and Agarwal, Sandhini and Ahmad, Lama and Akkaya, Ilge and Aleman, Florencia Leoni and Almeida, Diogo and Altenschmidt, Janko and Altman, Sam and Anadkat, Shyamal and others},
    note = {\textit{arXiv:2303.08774}},
    year={2023}
}

@INPROCEEDINGS{bai2024federated,
    author    = {Bai, Jiamu and Chen, Daoyuan and Qian, Bingchen and Yao, Liuyi and Li, Yaliang},
    booktitle = {Proc. Adv. Neural Inf. Process. Syst. (NeurIPS)},
    title     = {Federated Fine-Tuning of Large Language Models under Heterogeneous Tasks and Client Resources},
    year      = {2024},
    volume    = {37},
    number    = {},
    pages     = {14457--14483},
    month     = {Dec.},
    address   = {Vancouver, Canada}
}

@article{lin2024splitlora,
    title={Splitlora: A split parameter-efficient fine-tuning framework for large language models},
    author={Lin, Zheng and Hu, Xuanjie and Zhang, Yuxin and Chen, Zhe and Fang, Zihan and Chen, Xianhao and Li, Ang and Vepakomma, Praneeth and Gao, Yue},
    note ={\textit{arXiv:2407.00952}},
    year={2024}
}

@article{10835069,
    author={Qu, Guanqiao and Chen, Qiyuan and Wei, Wei and Lin, Zheng and Chen, Xianhao and Huang, Kaibin},
    journal={IEEE Commun. Surv. Tutorials}, 
    title={Mobile Edge Intelligence for Large Language Models: A Contemporary Survey}, 
    year={2025},
    volume={},
    number={},
    pages={1-1},
    month={early access, Jan. 09},
    doi={10.1109/COMST.2025.3527641},
    note={doi:{\color{blue}\href{http://dx.doi.org/10.1109/COMST.2025.3527641}{10.1109/COMST.2025.3527641}}}
}

@INPROCEEDINGS{Cai2023Efficient,
    author    = {Cai, Dongqi and Wu, Yaozong and Wang, Shangguang and Lin, Felix Xiaozhu and Xu, Mengwei},
    booktitle = {Proc. 29th Annu. Int. Conf. Mobile Comput. Netw. (MobiCom '23)},
    title     = {Efficient Federated Learning for Modern NLP},
    year      = {2023},
    volume    = {},
    number    = {},
    pages     = {1-16},
    month     = {Oct.},
    address   = {Madrid, Spain},
    doi       = {10.1145/3570361.3592505}
}

@inproceedings{mcmahan2017communication,
     author    = {McMahan, H. Brendan and Moore, Eider and Ramage, Daniel and Hampson, Seth and {Ag{\"u}era y Arcas}, Blaise},
      booktitle = {Proc. 20th Int. Conf. Artif. Intell. Stat. (AISTATS)},
      title     = {Communication-Efficient Learning of Deep Networks from Decentralized Data},
      year      = {2017},
      volume    = {54},
      number    = {},
      pages     = {1273-1282},
      month     = {Apr.},
      address   = {Fort Lauderdale, FL, USA},
      organization = {PMLR}
}

@inproceedings{li-etal-2025-mobilora,
    author    = {Li, Borui and Wang, Yitao and Ma, Haoran and Chen, Ligeng and Xiao, Jun and Wang, Shuai},
    booktitle = {Proc. 63rd Annu. Meeting Assoc. Comput. Linguist. (ACL) (Vol. 1: Long Papers)},
    title     = {MobiLoRA: Accelerating LoRA-based LLM Inference on Mobile Devices via Context-aware KV Cache Optimization},
    year      = {2025},
    volume    = {},
    number    = {},
    pages     = {23400-23410},
    month     = {Jul.},
    address   = {Vienna, Austria},
    doi       = {10.18653/v1/2025.acl-long.1140}
}

@article{vepakomma2018split,
  title={Split learning for health: Distributed deep learning without sharing raw patient data},
  author={Vepakomma, Praneeth and Gupta, Otkrist and Swedish, Tristan and Raskar, Ramesh},
  note={\textit{arXiv:1812.00564}},
  year={2018}
}

@inproceedings{thapa2022splitfed,
  author    = {Thapa, Chandra and Mahawaga Arachchige, Pathum Chamikara and Camtepe, Seyit and Sun, Lichao},
  booktitle = {Proc. AAAI Conf. Artif. Intell. (AAAI)},
  title     = {SplitFed: When Federated Learning Meets Split Learning},
  year      = {2022},
  volume    = {36},
  number    = {8},
  pages     = {8485-8493},
  month     = {Jun.},
  address   = {Vancouver, BC, Canada},
  doi       = {10.1609/aaai.v36i8.20825}
}

@ARTICLE{10839234,
  author={Qiang, Xianke and Chang, Zheng and Ye, Chaoxiong and Hämäläinen, Timo and Min, Geyong},
  journal={IEEE Wireless Commun.}, 
  title={Split Federated Learning Empowered Vehicular Edge Intelligence: Concept, Adaptive Design, and Future Directions}, 
  year={2025},
  volume={32},
  number={4},
  pages={90-97},
  doi={10.1109/MWC.009.2400219}}

@article{qiang2025deploying,
	title={Deploying Large AI Models on Resource-Limited Devices with Split Federated Learning},
	author={Qiang, Xianke and Liu, Hongda and Zhang, Xinran and Chang, Zheng and Liang, Ying-Chang},
	note={\textit{arXiv:2504.09114}},
	year={2025}
}

@ARTICLE{9611373,
    author={Wang, Yanmeng and Xu, Yanqing and Shi, Qingjiang and Chang, Tsung-Hui},
    journal={IEEE J. Sel. Areas Commun.}, 
    title={Quantized Federated Learning Under Transmission Delay and Outage Constraints}, 
    year={2022},
    volume={40},
    number={1},
    pages={323-341},
    month={Nov.},
    doi={10.1109/JSAC.2021.3126081}
}

@inproceedings{liu2023communication,
    title={Communication-efficient federated learning for heterogeneous edge devices based on adaptive gradient quantization},
    author={Liu, Heting and He, Fang and Cao, Guohong},
    booktitle={Proc. IEEE Conf. Comput. Commun. (INFOCOM)},
    pages={1--10},
    year={2023},
    address={New York City, NY, USA},
    month = {May},
    organization={IEEE}
}

@ARTICLE{10054381,
    title={On the Road to {6G}: Visions, Requirements, Key Technologies, and Testbeds}, 
    author={Wang, Cheng-Xiang and You, Xiaohu and Gao, Xiqi and Zhu, Xiuming and Li, Zixin and Zhang, Chuan and Wang, Haiming and Huang, Yongming and Chen, Yunfei and Haas, Harald and Thompson, John S. and Larsson, Erik G. and Renzo, Marco Di and Tong, Wen and Zhu, Peiying and Shen, Xuemin and Poor, H. Vincent and Hanzo, Lajos},
    journal={IEEE Commun. Surveys Tuts.}, 
    year={2023},
    volume={25},
    number={2},
    month = {Feb.},
    pages={905-974},
    publisher={{IEEE}},
    doi={10.1109/COMST.2023.3249835}}

@inproceedings{zeng2022not,
      author    = {Zeng, Wang and Jin, Sheng and Liu, Wentao and Qian, Chen and Luo, Ping and Ouyang, Wanli and Wang, Xiaogang},
      booktitle = {Proc. IEEE/CVF Conf. Comput. Vis. Pattern Recognit. (CVPR)},
      title     = {Not All Tokens Are Equal: Human-Centric Visual Analysis via Token Clustering Transformer},
      year      = {2022},
      volume    = {},
      number    = {},
      pages     = {11101-11111},
      month     = {Jun.},
      address   = {New Orleans, LA, USA},
      doi       = {10.1109/CVPR52688.2022.01082}
}

@ARTICLE{9923620,
    author={Liu, Xiaolan and Deng, Yansha and Mahmoodi, Toktam},
    journal={IEEE Trans. Wireless Commun.}, 
    title={Wireless Distributed Learning: A New Hybrid Split and Federated Learning Approach}, 
    year={2023},
    volume={22},
    number={4},
    pages={2650-2665},
    month={Oct.},
    doi={10.1109/TWC.2022.3213411}
}

@article{wu2023split,
    title={Split learning over wireless networks: Parallel design and resource management},
    author={Wu, Wen and Li, Mushu and Qu, Kaige and Zhou, Conghao and Shen, Xuemin and Zhuang, Weihua and Li, Xu and Shi, Weisen},
    journal={IEEE J. Sel. Areas Commun.},
    volume={41},
    number={4},
    pages={1051--1066},
    year={2023},
    month={Feb.},
    publisher={IEEE}
}

@ARTICLE{10714368,
    author={Qiang, Xianke and Chang, Zheng and Hu, Yun and Liu, Lei and Hämäläinen, Timo},
    journal={IEEE Internet Things J.	}, 
    title={Adaptive and Parallel Split Federated Learning in Vehicular Edge Computing}, 
    year={2025},
    volume={12},
    number={5},
    pages={4591-4604},
    month={Mar.},
    doi={10.1109/JIOT.2024.3479158}
}

@article{ma2025splitfrozen,
  title={SplitFrozen: Split Learning with Device-side Model Frozen for Fine-Tuning LLM on Heterogeneous Resource-Constrained Devices},
  author={Ma, Jian and Lyu, Xinchen and Jiang, Jun and Cui, Qimei and Yao, Haipeng and Tao, Xiaofeng},
  note={\textit{arXiv:2503.18986}},
  year={2025}
}

@ARTICLE{9277666,
    author={Zheng, Sihui and Shen, Cong and Chen, Xiang},
    journal={IEEE J. Sel. Areas Commun.	s}, 
    title={Design and Analysis of Uplink and Downlink Communications for Federated Learning}, 
    year={2021},
    month={Jul},
    volume={39},
    number={7},
    pages={2150-2167},
    doi={10.1109/JSAC.2020.3041388}
}

@ARTICLE{10256151,
    author={Wang, Bin and Fang, Jun and Li, Hongbin and Zeng, Bing},
    journal={IEEE Trans. Signal Process.}, 
    title={Communication-Efficient Federated Learning: A Variance-Reduced Stochastic Approach With Adaptive Sparsification}, 
    year={2023},
    volume={71},
    number={},
    pages={3562-3576},
    month={Sep.},
    doi={10.1109/TSP.2023.3316588}
}

@ARTICLE{10038639,
    author={Chen, Rui and Shi, Dian and Qin, Xiaoqi and Liu, Dongjie and Pan, Miao and Cui, Shuguang},
    journal={IEEE J. Sel. Areas Commun.}, 
    title={Service Delay Minimization for Federated Learning Over Mobile Devices}, 
    year={2023},
    month={Feb.},
    volume={41},
    number={4},
    pages={990-1006},
    doi={10.1109/JSAC.2023.3242711}
}

@INPROCEEDINGS{10621361,
  author    = {Su, Xiaoxin and Zhou, Yipeng and Cui, Laizhong and Lui, John C. S. and Liu, Jiangchuan},
  booktitle = {Proc. IEEE Conf. Comput. Commun. (INFOCOM)},
  title     = {Fed-CVLC: Compressing Federated Learning Communications with Variable-Length Codes},
  year      = {2024},
  volume    = {},
  number    = {},
  pages     = {601-610},
  month     = {May},
  address   = {Vancouver, Canada},
  doi       = {10.1109/INFOCOM52122.2024.10621361}
}

@article{10542529,
    author={Zhang, Xinran and Chen, Weilong and Zhao, Hui and Chang, Zheng and Han, Zhu},
    journal={IEEE Internet Things J.}, 
    title={Joint Accuracy and Latency Optimization for Quantized Federated Learning in Vehicular Networks}, 
    year={2024},
    volume={11},
    number={17},
    pages={28876-28890},
    month={Sept.},
    doi={10.1109/JIOT.2024.3406531}}

@ARTICLE{10026255,
    author={Lin, Xiaohan and Liu, Yuan and Chen, Fangjiong and Ge, Xiaohu and Huang, Yang},
    journal={IEEE Trans. Green Commun. Networking}, 
    title={Joint Gradient Sparsification and Device Scheduling for Federated Learning}, 
    year={2023},
    volume={7},
    number={3},
    pages={1407-1419},
    month={Jan.},
    doi={10.1109/TGCN.2023.3239373}
}

@ARTICLE{10791300,
  author={Zhang, Junhe and Ni, Wanli and Wang, Dongyu},
  journal={IEEE Trans. Veh. Technol.}, 
  title={Federated Split Learning With Model Pruning and Gradient Quantization in Wireless Networks}, 
  year={2025},
  volume={74},
  number={4},
  pages={6850-6855},
  doi={10.1109/TVT.2024.3515083}}

@article{zheng2023reducing,
  title={Reducing communication for split learning by randomized top-k sparsification},
  author={Zheng, Fei and Chen, Chaochao and Lyu, Lingjuan and Yao, Binhui},
  journal={\textit{arXiv:2305.18469}},
  year={2023}
}

@ARTICLE{10538233,
    author={Wang, Lingyi and Wu, Wei and Zhou, Fuhui and Yang, Zhaohui and Qin, Zhijin and Wu, Qihui},
    journal={IEEE Trans. Commun.}, 
    title={Adaptive Resource Allocation for Semantic Communication Networks}, 
    year={2024},
    volume={72},
    number={11},
    month={May.},
    pages={6900-6916},
    doi={10.1109/TCOMM.2024.3405355}
}

@article{hu2024energy,
    title={Energy-efficient federated edge learning with streaming data: A lyapunov optimization approach},
    author={Hu, Chung-Hsuan and Chen, Zheng and Larsson, Erik G.},
    journal={IEEE Transactions on Communications}, 
    title={Energy-Efficient Federated Edge Learning With Streaming Data: A Lyapunov Optimization Approach}, 
    year={2025},
    volume={73},
    number={2},
    pages={1142-1156},
    month={Feb.},
    doi={10.1109/TCOMM.2024.3443731}
}

@ARTICLE{11045879,
    author={Qiang, Xianke and Chang, Zheng and Min, Geyong},
    journal={IEEE Trans. Mob. Comput.}, 
    title={AIGC-Assisted Federated Learning for Vehicular Edge Intelligence: Vehicle Selection, Resource Allocation and Model Augmentation}, 
    year={2025},
    volume={24},
    number={11},
    month={Jun.},
    pages={11896-11909},
    doi={10.1109/TMC.2025.3581983}
}

@article{russakovsky2015imagenet,
  author    = {Russakovsky, Olga and Deng, Jia and Su, Hao and Krause, Jonathan and Satheesh, Sanjeev and Ma, Sean and Huang, Zhiheng and Karpathy, Andrej and Khosla, Aditya and Bernstein, Michael and Berg, Alexander C. and Fei-Fei, Li},
  journal   = {Int. J. Comput. Vis.},
  title     = {ImageNet Large Scale Visual Recognition Challenge},
  year      = {2015},
  volume    = {115},
  number    = {3},
  pages     = {211-252},
  month     = {Dec.},
  doi       = {10.1007/s11263-015-0816-y}
}

@inproceedings{nilsback2008automated,
 author    = {Nilsback, Maria-Elena and Zisserman, Andrew},
  booktitle = {Proc. 6th Indian Conf. Comput. Vis., Graph. Image Process. (ICVGIP)},
  title     = {Automated Flower Classification over a Large Number of Classes},
  year      = {2008},
  volume    = {},
  number    = {},
  pages     = {722-729},
  month     = {Dec.},
  address   = {Bhubaneswar, India},
  doi       = {10.1109/ICVGIP.2008.47},
  organization = {IEEE}
}

@TECHREPORT{wah2011caltech,
   author      = {Wah, Catherine and Branson, Steve and Welinder, Peter and Perona, Pietro and Belongie, Serge},
  title       = {The Caltech-UCSD Birds-200-2011 Dataset},
  institution = {California Institute of Technology},
  year        = {2011},
  number      = {CNS-TR-2011-001},
  address     = {Pasadena, CA, USA}
}
	
\end{document}